\documentclass[onecolumn,journal,11pt]{IEEEtran}
\usepackage{etex}
 
\usepackage{graphicx}
\usepackage[square, numbers, comma, sort&compress]{natbib}  
\usepackage{amsmath,amsfonts,amssymb,amscd,amsthm,xspace}
\usepackage{ifpdf}
\usepackage{epstopdf}
\usepackage{array}
\usepackage{lettrine}
\usepackage[scriptsize]{subfigure}
\graphicspath{{./Figures/}}
\usepackage{setspace}
\usepackage{epstopdf}
\usepackage{authblk}
\usepackage{xcolor} 
\usepackage[kerning=true]{microtype} 
\usepackage{float}  
\usepackage{hyperref}
\usepackage{fancyhdr}
\usepackage{booktabs}
\usepackage{cleveref}
\usepackage{algorithm}
\usepackage{algpseudocode}
\usepackage{lipsum}

\usepackage{tgpagella}

\newcommand{\li}{n}  
\newcommand{\lo}{m}			
\newcommand{\mm}{\mathcal{A}}    
\newcommand{\ivt}{\mathbf{X}}   
\newcommand{\ovt}{\mathbf{Y}}	
\newcommand{\rvt}{\hat{\ivt}}  
\newcommand{\ep}{P_{\textnormal{e}}^{\textnormal{blk}}}     
\newcommand{\er}{E^{blk}}    
\newcommand{\ent}{H}   
\newcommand{\mif}{I}   
\newcommand{\real}{\mathbb{R}}  
\newcommand{\precision}{Q}  

\newcommand{\nr}{\ell_{q}}
\newcommand{\up}{U}

\newtheorem{theorem}{Theorem}
\newtheorem{definition}{Definition}
\newtheorem{lemma}{Lemma}
\newtheorem{corollary}{Corollary}

\newcommand{\cmt}[1]{}

\makeatletter
\newcommand{\xRightarrow}[2][]{\ext@arrow 0359\Rightarrowfill@{#1}{#2}}
\newcommand{\xLeftrightarrow}[2][]{\ext@arrow 0359\Leftrightarrowfill@{#1}{#2}}
\newcommand\myeqa{\mathrel{\overset{\makebox[0pt]{\mbox{\normalfont\tiny\sffamily (a)}}}{\geq}}}
\newcommand\myeqb{\mathrel{\overset{\makebox[0pt]{\mbox{\normalfont\tiny\sffamily (b)}}}{\geq}}}
\makeatother

\begin{document}

\title{Energy-efficient Decoders for Compressive Sensing: Fundamental Limits and Implementations}

\author{Tongxin Li${}^1$\qquad\qquad Mayank Bakshi${}^1$\qquad\qquad Pulkit Grover${}^2$\vspace{0.5em}\\ 
\small ${}^1$The Chinese University of Hong Kong\qquad${}^2$Carnegie Mellon University}\maketitle

\IEEEpeerreviewmaketitle

\begin{abstract}
The fundamental problem considered in this paper is ``What is the \textit{energy} consumed for the implementation of a \emph{compressive sensing} decoding algorithm on a circuit?". Using the ``{\em information-friction}'' framework introduced in~\cite{grover2013information}, we examine the smallest amount of \textit{bit-meters}\footnote{The \lq\lq{}\textit{bit-meters}\rq\rq{} metric was first proposed in \cite{grover2013information}, as an alternative to the VLSI model introduced by Thompson and others in~\cite{thompson1979area,thompson1980complexity, brent1981area, chazelle1981towards, leiserson1981area, mead1980introduction} (and explored further in~\cite{sinha1988new, kramer1984vlsi, bhatt2008area, cole1988optimal, thompson1981vlsi}) for measuring the energy consumed in a circuit.} as a measure for the energy consumed by a circuit. We derive a fundamental lower bound for the implementation of compressive sensing decoding algorithms on a circuit. In the setting where the number of measurements scales linearly with the sparsity and the sparsity is sub-linear with the length of the signal, we show that the \textit{bit-meters} consumption for these algorithms is order-tight, {\em i.e.}, it matches the lower bound asymptotically up to a constant factor. Our implementations yield interesting insights into design of energy-efficient circuits that are not captured by the notion of computational efficiency alone. 
\end{abstract} %

\section{Introduction}
\lettrine[lines=1]{C}{ompressive Sensing} has emerged as an attractive paradigm in recent years~\cite{candes2006robust,donoho2006compressed}. Motivated by applications where processing the dataset without exploiting the underlying sparsity is prohibitively expensive, compressive sensing aims to reduce the cost of processing through algorithms that take sparsity into account. Initial work on compressive sensing showed that the number of measurements required to sketch a signal of length $n$ and sparsity $k$ is  ${\cal O}(k\log n)$\cite{candes2006robust,donoho2006compressed}. Subsequently, computationally efficient algorithms for this problem have also been discovered~\cite{Can:08,BarDDW:08,BerIR:08,BerI:09,GilI:10,bakshi2012sho,PawR:12}. The fastest of these algorithms uses a peeling type decoder and have running time ${\cal O}(k)$ with ${\cal O}(k)$ measurements~\cite{bakshi2012sho,PawR:12}. 

In this paper, we adopt an energy-centric view of compressive sensing. Our motivation comes from applications such as ad-hoc wireless networks~\cite{Guo:10}, where decoding energy is of critical importance. In these applications, since the decoder of often a batter-powered device, processing the received measurements to obtain the desired reconstruction is a fundamentally limiting aspect of the system design. Notably, the computationally efficient algorithms of~\cite{bakshi2012sho,PawR:12} are no longer order-optimal when the decoding energy is the metric of interest. Therefore, we ask the question \emph{``What is the smallest amount of energy required to decode a signal from its compressed measurements?''}

As an exploratory work, we examine the problem in the \emph{information-friction} framework. This framework was introduced in~\cite{grover2013information} for finding a trade-off between the energy consumed in encoding/decoding processes and transmission power in a communication system. In practice, it's reasonable for us to relate {\it{bit-meters}} with energy consumed for the decoding process, as~\cite{grover2013information} elaborated through multiple different scenarios. We first show that for a fixed precision $\precision$, the required \textit{bit-meters} (energy) for decoding a compressed signal can be no smaller than $\Omega\left( \sqrt{\frac{\li k}{\log \li}}\sqrt{\log{\frac{1}{\ep}}}\right)$ asymptotically at the regime $m = \Theta\left(k\right)$ and $k= n^{1-\beta}$ where constant $\beta\in (0,1)$. We show that this asymptotic lower bound is order-tight by giving two multi-stage algorithms for each the \textit{bit-meters} is $\mathcal{O}\left(\sqrt{\frac{\li k}{\log{\li}}}\sqrt{\log{\frac{1}{\ep}}}\right)$.

The rest of the paper is organized as follows. We begin with describing our model in Section~\ref{sec:background}. The main results of the paper stated in Section~\ref{Main Results}. The key ideas in the proofs of these results are outlined in Section~\ref{sec:key}. Finally, the main parts of the proofs are described in Appendices~\ref{appendix:lb}-\ref{appendix:up}. 

\section{Background and Definitions}
\label{sec:background}
In this section, we formalize the models used in this paper.


\subsection{\em Compressive Sensing}
%

For compressive sensing, the input vector $\ivt\in\real^{\li}$ is a real-valued vector of length $\li$. The linear encoding process is represented by an \textit{encoding matrix} $\mm\in \mathbb{C}^{m\times n}$. The vector $\ovt=\mm\ivt\in\mathbb{C}^{\lo}$ of length $\lo$ is the corresponding \textit{output vector}. Based on $\ovt$, a \textit{recovery vector} $\rvt\in\real^{\li}$ is decoded using a decoding algorithm. For sparsity, we consider two basic models----probabilistic and combinatorial. We assume the length of \textit{input vector} is considerably large. Therefore, the two models are asymptotically equivalent\footnote{Note that by the strong law of large numbers (see, for instance the excellent textbook~\cite{grimmett1992probability}), or even a weaker statement by Chernoff bound in~\cite{chernoff1952measure}, the number of non-zero entries of the input vector $\ivt$ in the probabilistic sparsity sensing model will be bounded in a constant range containing $k$ (\textit{e.g.} $[k/2,3k/2]$) with an exponential probability of error $e^{-\Theta{k}}$, which is actually, negligible compared with the error probability we could achieve using the algorithms in Section~\ref{Upper Bound} with the upper bound on \textit{bit-meters} (see Definition~\ref{def:bitmeter}) provided in Theorem~\ref{thm:bitmeter} at Section~\ref{Main Results}. }. We 
adopt the probabilistic one with independent property to simplify calculations for the ease of analysis. For the purpose of presentation, the combinatorial model is used to give us a concise insight into both the lower and the upper bounds.
\begin{definition}[Sparsity Model ($\li$,$m$,$p$)]
\label{pcs}
A bounded length-$n$ ``compressible" vector $\ivt\in\real^{\li}$ is an {\textnormal{Input Vector}} whose each entry $X_i\in\real$ satisfies $|X_i|\leq\up$ (with a constant upper bound $\up\geq 0$) and has probability\footnote{We assume $p=o(1)$ as the {\it{sparse assumption}} for the {\it{Sparsity Model ($\li$,$m$,$p$)}}.} $p$ to be non-zero.
\end{definition}

\begin{definition}[Sparsity Model ($\li$,$m$,$k$)]
\label{def:csm}
A bounded length-$n$ ``compressible" vector $\ivt\in\real^{\li}$ is the $k$-sparse\footnote{We assume $k=o(n)$ as the {\it{sparse assumption}} for the {\it{Sparsity Model ($\li$,$m$,$k$)}}.} {\textnormal{Input Vector}} if it contains exactly $k$ non-zero entries $X_i\in\real$ satisfying $|X_i|\leq\up$ (with a constant upper bound $\up\geq 0$).
\end{definition}


In this paper, we focus on the asymptotic regime where $k=k(n)$ and $p=p(n)$ such that $k=\omega(1)\cap{o}(\li)$ and $p=\omega(1/\li)\cap{o}(1)$, \textit{i.e.}, both $k$ and $\li p$ grow sub-linearly with $\li$. Our theorem for upper bounds in section~\ref{Main Results} is restricted in the sub-linear regime $k= n^{1-\beta}$ where $\beta\in (0,1)$.

We define the \textit{average block error probability} based on quantization and a given \textit{norm} $|| \cdot||_{\nr}$ of interest for $0\leq q\leq\infty$.

\begin{definition}[Reconstruction Error, Precision, Average Block Error Probability]
Given the {\textnormal{Input Vector}} $\ivt$ and {\textnormal{Recovery Vector}} $\rvt$, the {\textnormal{Reconstruction Error}} is defined as $||\ivt-\rvt||_{\nr}$ and the {\textnormal{Relative Error}} is defined further as $||\ivt-\rvt||_{\nr}/||\ivt||_{\nr}$. Let $\precision$ denote the {\textnormal{Precision}}, i.e., the required number of bits for reconstructing the input vector $\ivt$, then the {\textnormal{Average Block Error Probability}} is defined by $\ep=\Pr(\er=1)$ where $\er=1$ if the relative error satisfies $||\ivt-\rvt||_{\nr}/||\ivt||_{\nr}> 2^{-\precision}$; Otherwise, $\er=0$.
\end{definition}

\subsection{\em Implementation Model}

A decoding circuit has two functionalities----storing the \textit{output vector} and processing it to obtain \textit{recovery vector}. 

We think of a decoding circuit as a "graph" whose nodes from a subset of \textit{points} of a two dimensional lattice $\Lambda\subset\real^{2}$. Each node can both store one real number and perform a computation. Nodes are connected through \textit{undirectional links} that represents the wiring in the circuit. Furthermore, each node on the lattice $\Lambda$ (\textit{i.e.}the circuit) has a constant \textit{packing radius} $\rho(\Lambda)>0$ to ensure sufficient distance between the nodes for real implementation\footnote{While we describe a more general framework here, for most of our results, it suffices to restrict our attention to square lattices.}. Considering the coherence with~\cite{grover2013information}, we define the generalized circuit model succinctly using the same order of definitions as~\cite{grover2013information} did. 


\begin{definition}[$\rho$--Lattice ]
A lattice $\Lambda$ is a {\textnormal{$\rho$--Lattice}} if it spans $\mathbb{R}^{2}$ and the {\textnormal{packing radius}} $\rho(\Lambda)>0$ is at least $\rho$.
\end{definition}

\begin{definition}[Substrate]
A {\textnormal{Substrate}} $V$ is a compact subset of $\real^{2}$.
\end{definition}

\begin{definition}[Grid ($\Lambda,V$)]
A {\textnormal{Grid ($\Lambda,V$)}} is a sub-lattice $\Lambda_{V}$ defined as the intersection of a {\textnormal{Lattice $\Lambda$}} and a {\textnormal{Substrate}} $V$ such that $\Lambda_{V}=\Lambda\cap V$.

\end{definition}

\begin{definition}[Decoding Circuit, Sub-Circuit, Computational Nodes, Input-nodes, Output-nodes]
The {\textnormal{Substrate}} $V$ together with a collection $S\subset\Lambda_{V}$ of points (called {\textnormal{Computational Nodes}}, or simply {\textnormal{Nodes}}) inside the {\textnormal{Grid $\Lambda_{V}$}}, is called a {\textnormal{Decoding Circuit}}, and is denoted by $DecCkt=(V,S)$. A {\textnormal{Sub-Circuit}} denoted by $SubCkt=(V_{\textnormal{ sub}},S_{\textnormal{ sub}})$ is a bounded subset $V_{\textnormal{ sub}}$ of $V$ together with a subset of {\textnormal{Computational Nodes}} $S_{\textnormal{ sub}}\subset \Lambda\cap V_{\textnormal{ sub}}$. The {\textnormal{Output-nodes}} are defined as the nodes for storing the Recovery vector $\rvt$, and the {\textnormal{Input-nodes}} are defined as the nodes for storing the received Output vector $\ovt$. Nodes can accomplish noiseless communication and with each other through undirectional {\textnormal{Links}}, within which binary strings of messages are transmitted. Nodes can also implement arithmetic calculations.
\end{definition}


Note that we asumme that for each node (both \textit{input-node} and \textit{output-node}), it stores exactly one entry\footnote{It is reasonable to assume the one-to-one correspondence between nodes and entries. Another possible model may be a one-to-one correspondence between bits and nodes which is a possible direction in the future.} of the corresponding vector. Moreover, we assume \textit{timing} is available for the \textit{output-nodes} which means only non-trivial data are required to be transmitted and if during a period of time no data was received, the \textit{output-nodes} are able to automatically declare the corresponding entries in the \textit{recovery vector} zeros\footnote{The assumption about \textit{timing} is critical. In fact, if no \textit{timing} is available, every \textit{output-node} is required to receive at least $1$ bit with a constant \textit{communication distance}, which implies a lower bound on \textit{bit-meters} $\Omega\left(\li\right)$. This violates the spirit of compressive sensing as we want the energy (\textit{bit-meters}) required to be sub-linear in $\li$ in light of the sparsity on the signal.}.

Moreover, we assume that each node, including \textit{input} and \textit{output} nodes can behave as a bridge for communication between other nodes. Thus, any computation that can be performed by an intermediate node can be, in principle, be performed by an input or output node. Therefore, exactly $n+m$ nodes are sufficient for any decoding circuit. We assume that the computational nodes can communicate noiselessly with each other through unidirectional {\textnormal{links}} as defined above.

%
\begin{definition}[Communication Distance]
	Given two nodes $\textbf{x}$ and $\textbf{y}$, the {\textnormal{Communication Distance}} $\mathcal{D}\left(\textbf{x},\textbf{y}\right)$ is the {\em Euclidean distance} $||\textbf{x}-\textbf{y}||_{2}$ between nodes $\textbf{x}$ and $\textbf{y}$.
\end{definition}


Let $S\subset\Lambda_{V}$ be a collection of nodes. The number of bits communicated between $\textbf{x}, \textbf{y}\in S$ is denoted by $\mathcal{B}\left(\textbf{x}, \textbf{y}\right)\geq 0$. We then define our fundamental \textit{measure} of energy below.

\begin{definition}[Bit-meters]
	\label{def:bitmeter}
	The {\textnormal{Bit-meters}} or $\mu\left(\cdot\right)$ is a non-negative real-valued {\textnormal{measure}} from the power set over collections of nodes $\mathcal{P}\left(S\right)$ to the {\textnormal{extended real number line}} satisfying the following properties:
	\begin{enumerate}
		\item $\mu(\emptyset)$ = 0.
		\item For every $\textbf{x}, \textbf{y}\in S$, $\mu\left(\textbf{x},\textbf{y}\right)=\mathcal{B}\left(\textbf{x}, \textbf{y}\right)\times {\mathcal{D}\left(\textbf{x},\textbf{y}\right)}$.
		\item For every subset $E\subset S$, $\mu\left(E\right)=\sum_{\textbf{x},\textbf{y}\in E}{\mathcal{B}\left(\textbf{x}, \textbf{y}\right)\times {\mathcal{D}\left(\textbf{x},\textbf{y}\right)}/2}$.
		\item For every {\textnormal{Decoding Circuit}} (or {\textnormal{Sub-circuit}}) $DecCkt=\left(V,S\right)$, $\mu\left(DecCkt\right)=\mu\left(S\right)$. 
	\end{enumerate}
\end{definition}

Our main goals are to obtain asymptotic lower and upper bounds on {\it{bit-meters}} for decoding circuits implementing compressive sensing algorithms, \textit{i.e.}, find the bounds on $\mu\left(DecCkt\right)$ as $\li\rightarrow\infty$. Note that all the discussions are based on the implementation model summarized below.

\begin{definition}[Implementation Model $(\rho$,$\mu)$ ]
\label{IM}
{\textnormal{Implementation Model $(\rho, \mu)$}} denotes the pair of {\textnormal{Decoding Circuit}} $DecCkt$ defined on a {\textnormal{$\rho$--Lattice}} and {\textnormal{Bit-meters}} $\mu$. 
\end{definition}

\section{Main Results}
\label{Main Results}

Let $R=\lo/\li$ be the rate of compressive sensing. As our first result, we obtain the following general lower bound on \textit{bit-meters}.

\begin{theorem}[Lower Bound]
\label{thm:lower bound}
Consider the {\textnormal{Sparsity Model $(n,m,p)$}}. For any {\textnormal{encoding matrix}} $\mathcal{A}$ and {\em decoding circuit} $DecCkt$  implemented on the {\textnormal{Implementation Model $(\rho,\mu)$}}, we have:
\begin{enumerate}
	\item The {\textnormal{average block error probability}} $\ep\geq 0$.
	\item The {\textnormal{bit-meters}} 
	\begin{align}
	\label{lower bound}
	&\mu{(DecCkt)} \geq \frac{\rho C_{0}}{24\sqrt{2}}\li\left(\frac{1}{4}-R\right)
	p\precision\sqrt{\left(\frac{1}{2R}-\frac{1-R^2+2R}{(1+R)^2}\right)}\sqrt{\log_p{10\ep}}.
	\end{align}
\end{enumerate}
\end{theorem}

As a consequence of the above theorem, we derive the following corollary that states the asymptotic scaling of the lower bound with respect to $\li$. This serves as a benchmark for our algorithm design subsequently.

\begin{corollary}[Scaling for Lower Bound]
	\label{corollary}
	Consider the {\textnormal{Sparsity Model $(n,m,k)$}}. Assume the {\textnormal{precision}} $\precision=\Theta\left(1\right)$. For any {\textnormal{encoding matrix}} $\mathcal{A}$ and {\em decoding circuit} $DecCkt$  implemented on the {\textnormal{Implementation Model $(\rho,\mu)$}}, we have:
	\begin{enumerate}
		\item The {\textnormal{average block error probability}} $\ep\geq 0$.
		\item The {\textnormal{bit-meters}} 
	\begin{align*}
	&\mu{(DecCkt)} = \Omega\left( \sqrt{\frac{ k^2}{R\log \li}}\min\left(\sqrt{k},\sqrt{\log{\frac{1}{\ep}}}\right)\right).
	\end{align*}
	\end{enumerate}
\end{corollary}

For the regime $\lo=\Theta(k)$, $k= n^{1-\beta}$ where $\beta\in (0,1)$, and $\ep=e^{-\Omega\left(\lo\right)}$ of our interests, we derive the following upper bound which matches the order of \textit{bit-meters} in Corollary~\ref{corollary}.

\begin{theorem}[Upper Bound]
	\label{thm:bitmeter}
	Consider the {\textnormal{Sparsity Model $(\li,\lo,k)$}}. Let $k= n^{1-\beta}$ for some $\beta\in (0,1)$, there exist an {\textnormal{encoding matrix}} $\mathcal{A}$ and a {\em decoding circuit} $DecCkt$ implemented on the {\textnormal{Implementation Model $(\rho,\mu)$}} such that
	\begin{enumerate}
		\item The {\textnormal{average block error probability}}  $\ep=\mathcal{O}\left(1/\sqrt{k}\right)$.
		\item The {\textnormal{number of measurements}} $m=\Theta(k)$.
		\item The {\textnormal{precision}} $\precision=\Theta\left(1\right)$.
		\item The {\textnormal{bit-meters}}  $\mu\left(DecCkt\right)=\mathcal{O}\left(\sqrt{\li k}\right)$.
	\end{enumerate}
\end{theorem}

As a consequence of Corollary~\ref{corollary} and the above upper bound, we state the following corollary as a conclusion.

\begin{corollary}[Order-tight Bound]
	\label{coro:tight}
		Consider the {\textnormal{Sparsity Model $(\li,\lo,k)$}}. Let $k= n^{1-\beta}$ for some $\beta\in (0,1)$ and $\lo=\Theta\left(k\right)$. There exist an {\textnormal{encoding matrix}} $\mathcal{A}$ and a {\em decoding circuit} $DecCkt$ implemented on the {\textnormal{Implementation Model $(\rho,\mu)$}} such that 
	\begin{align*}
	&\mu{(DecCkt)} = \Theta\left( \sqrt{\frac{\li k}{\log\li}}\sqrt{\log{\frac{1}{\ep}}}\right).
	\end{align*}
\end{corollary}

In the next section, we give an overview of the proofs. The detailed proofs can be found in Appendices~\ref{appendix:lb}-\ref{appendix:up}.

\section{Main Ideas}
\label{sec:key}
\subsection{\em Lower Bound}
\label{Lower Bound}

In this section, we describe the main ideas from the derivation of the lower bound. Using the "Stencil-partition" idea introduced by~\citep{grover2013information}, we divide the entire circuit into several sub-circuits\footnote{Note that since we define the circuit using \textit{lattice} framework, it natural to define \textit{Stencil-partition} using \textit{sub-lattice}. Hence we call each sub-circuit "Parallelepiped" sometimes in the lemmas at Appendix~\ref{appendix:lb}} and find the minimal number of bits communicated between each sub-circuit.


\begin{figure}
	\centering
	\includegraphics[scale=0.35]{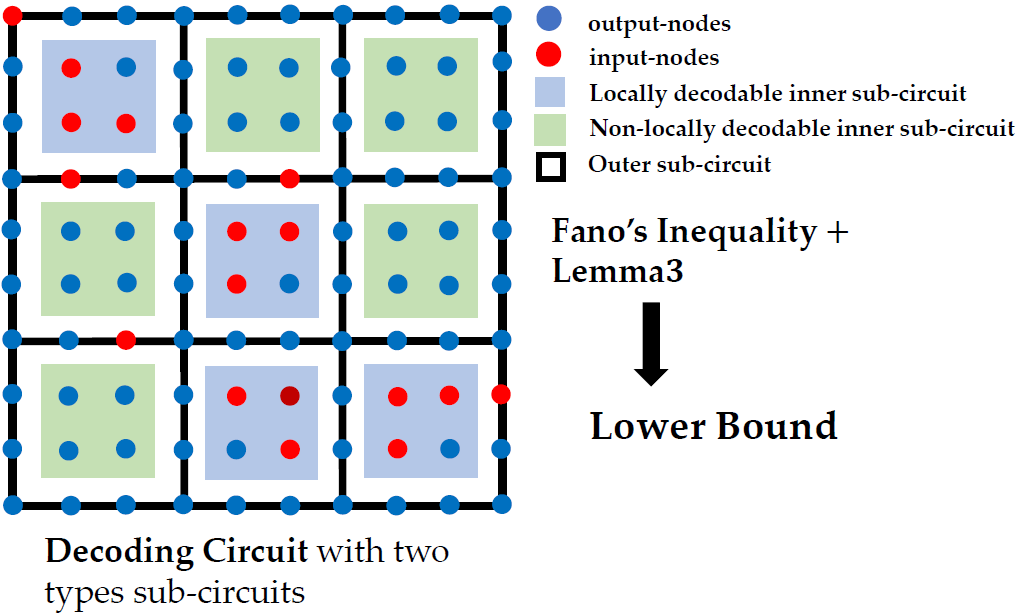}
	\caption{This conceptual graph illustrates the idea of the derivation of lower bound. The figure contains two types of sub-circuits: \textit{locally decodable} sub-circuits and \textit{non-locally decodable} sub-circuits. In details, we have the number of sub-circuits $L=9$ and $m=17$ thus $2m/L=38/9$ and using Lemma~\ref{lemma_2} the classifications of sub-circuits types are provided in this figure using different colors. As a result, the {\em bit-meters} are bounded from below as Theorem~\ref{thm:lower bound} states.}
	\label{fig:lemma}
\end{figure}

\begin{definition}[Stencil-partition]
	\label{def:sp}
	For point any $\textbf{u}\in\Lambda$, a {\textnormal{Stencil$(\lambda,\eta,\textbf{u})$}} on  {\textnormal{Implementation Model$(\rho(\Lambda),\mu)$}} consists of the following: 
	\begin{enumerate}
	\item A {\textnormal{Sub-lattice $\Lambda_{0}\subset\Lambda$}} with order of {\textnormal{quotient}} $|\Lambda/\Lambda_{0}|=\lambda$.
	\item The {\em outer parts} of sub-circuits induced by the {\textnormal{cosets}} $\textbf{u}+\Lambda_{0}=\lbrace{\textbf{u}+\textbf{v}:\textbf{v}\in\Lambda_{0}}\rbrace$.
	\item The {\em inner parts} of sub-circuits induced by scaling each {\em outer part} using a {\em fractional parameter} $\eta$.
	\end{enumerate}
	Let the $i$-th sub-circuit have $m_i$ {\textnormal{input-nodes}} and $n_i$ {\textnormal{output-nodes}}  within the {\em outer part} of sub-circuit. Let $i$-th sub-circuit have $m_i^{{\textnormal{inside}}}$ {\textnormal{input-nodes}} and $n_i^{{\textnormal{inside}}}$ {\textnormal{output-nodes}} within the {\em inner part} of sub-circuit\footnote{If any computational node lies on the boundary of two {\em outer parts} of sub-circuits, then it is arbitrarily included in one of them.}.
\end{definition}

Figure~\ref{fig:lemma} shows the geometric ideas. We first use a \textit{stencil} to divide the decoding circuit into several sub-circuits. Each sub-circuit consists of an \textit{inner part} and \textit{outer bound} (see Figure~\ref{fig:stencil} for more details). Next, based on the ratio of numbers of \textit{input-nodes} and \textit{output-nodes} inside the sub-circuit, we define two types of sub-circuits: \textit{locally decodable} sub-circuits and \textit{non-locally decodable} sub-circuits. Then we argue that the fraction of inner sub-circuits whose \textit{output-nodes} can be fully decoded using the \textit{input-nodes} within itself is a constant smaller than one. We mainly focus on the second type of sub-circuits, since these sub-circuits do not have enough information to fully decode all \textit{output-nodes} from the \textit{input-nodes} within the sub-circuits. By using Fano's inequality~\citep{cover2012elements} we finally deduce that the \textit{inner parts} of \textit{non-locally decodable} sub-circuits must communicate with other sub-circuits, giving a bound on \textit{bit-meters} stated in Theorem~\ref{thm:lower bound}.


\subsection{\em Upper Bounds}
\label{Upper Bound}

In this section, we explore different constructions for the implementation-circuits of compressive sensing algorithms. The basic issues here are the locations of the two types of nodes (input and output) and how they communicate with each other. We consider two types of algorithms----algorithms with centrally located \textit{input-nodes}, and algorithms involving distributed arrangement of nodes. 



 In our regime of interests, \textit{i.e.,} $m = \Theta(k)$, the later design always dominates the former one. This gives us the insights that local-decoding helps significantly in reducing the energy consumed and approaching an order-optimal performance.
 
 \subsubsection{Centralized-Decoding Algorithms}
 
\begin{figure}
	\centering
		\includegraphics[scale=0.35]{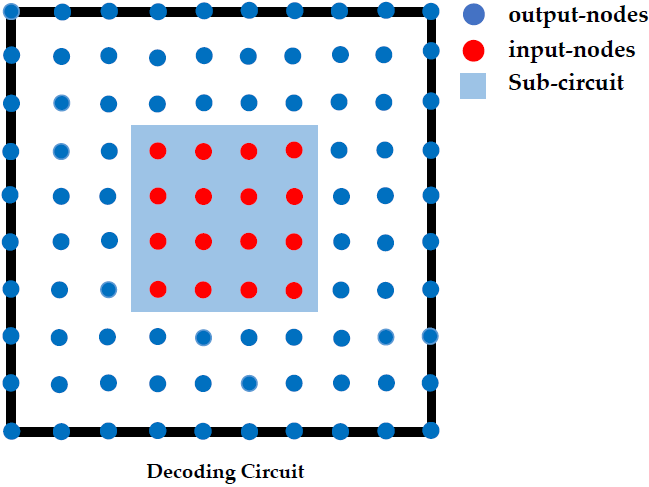}
		\caption{This conceptual graph illustrates one possible centralized arrangement of {\it{input-nodes}} and \textit{output-nodes} on a implementation circuit, whereby the lower bound on {\it{bit-meters}} could be derived. Within the decoding circuit, all $k$ \textit{input-nodes} are located in the central part and all $\li$ \textit{output-nodes} are put in the surrounding area of the central part.}
		\label{al:cent}
\end{figure}

Centralization of \textit{input-nodes} is perhaps the simplest construction possible. An interesting intuition here is that this design is better for those algorithms that have a relatively higher decoding complexity and lager number of measurements and thus, require a lager number of \textit{input-nodes} talk to each other frequently. For such algorithms, centralization improves the performance in terms of the error probability by enabling greater cooperations between nodes. Figure~\ref{al:cent} shows the idea of centralization.

However, for algorithms with relatively sparse \textit{encoding matrices} such that the average block error probability $\ep$ is of the order $\ep=e^{-\Omega{(m)}}$, we claim that
the centralized design involves a ``gap'' between the \textit{bit-meters} consumed and the scaling lower bound stated in Corollary~\ref{corollary} as the following arguments indicate. 


Consider the {\textit{Sparsity Model $(n,m,k)$}}, for any \textit{encoding matrix} $\mathcal{A}$ and decoding circuit $DecCkt$ implemented on the {\textit{Implementation Model $(\rho,\mu)$}}, assume precision $\precision=\Theta\left(1\right)$. Let $S_{\textnormal{Input}}$ be the set containing all \textit{input-nodes} and $S_{\textnormal{Output}}$ be the set containing all \textit{output-nodes}. If we centralize the {\em input-nodes} to make sure that for any $\textbf{x}\in S_{\textnormal{Input}}$, there is a positive $\rho=\Theta\left({\sqrt{m}}\right)$ such that the ball with {\em packing radius} $\rho_0$ contains $S_{\textnormal{Input}}$. Since there are $\li$ {\em output-nodes}, on average, chosen $\textbf{x}\in S_{\textnormal{Input}}$ and $\textbf{y}\in S_{\textnormal{Output}}$ uniformly at random, the expected {\em communication distance} ${\mathbb{E}[\mathcal{D}\left(\textbf{x},\textbf{y}\right)}]=\Theta{\left(\sqrt{\li}\right)}$. Since at least $\Theta\left(k\right)$ bits of information have to be transmitted from {\em input-nodes} to {\em output-nodes}, we obtain a lower bound on {\em bit-meters} $\mu\left(DecCkt\right)=\Omega\left(k\sqrt{\li}\right)$. Now based on our Corollary~\ref{corollary}, there are two cases----{\em average block error probability} $\ep=e^{-\Omega\left(\lo\right)}$ and $\ep=e^{-\mathcal{O}\left(\lo\right)}$. For the first case, we have $\mu{(DecCkt)} = \Omega\left( k\sqrt{\frac{ \li k}{\lo\log \li}}\right)$; for the second case, we have $\mu{(DecCkt)} = \Omega\left( k\sqrt{\frac{\li\lo}{\log \li}}\sqrt{\log{\frac{1}{\ep}}}\right)$ which implies $\mu{(DecCkt)} = \Omega\left( k\sqrt{\frac{\li}{\log \li}}\right)$. Therefore, for both cases, we conclude that the centralization of {\em input-nodes} is not able to achieve an order-tight upper bound on {\em bit-meters}.

%
%

 \subsubsection{Distributive-Decoding Algorithm}

We propose two energy-efficient compressive sensing algorithms shown in Figure~\ref{fig:dis}. Both use the idea of local decoding to reduce the \textit{bit-meters} required. Instead of arranging all the \textit{input-nodes} in the central part of the circuit, we distribute them throughout the circuit with carefully designed algorithms. As an intuition, since a large fraction of the communication is carried only in a small region, the consumed energy is reduced significantly. Leaving the formal definitions for Section~\ref{sec:ad} and Appendix~\ref{appendix:da}. We give brief descriptions of the algorithms in the next section~\ref{sec:ad} using a stage-by-stage manner accompanied by schematic graphs. The analysis of the performance is given Appendix~\ref{appendix:up}.

\begin{figure}
	\centering
	\includegraphics[scale=0.25]{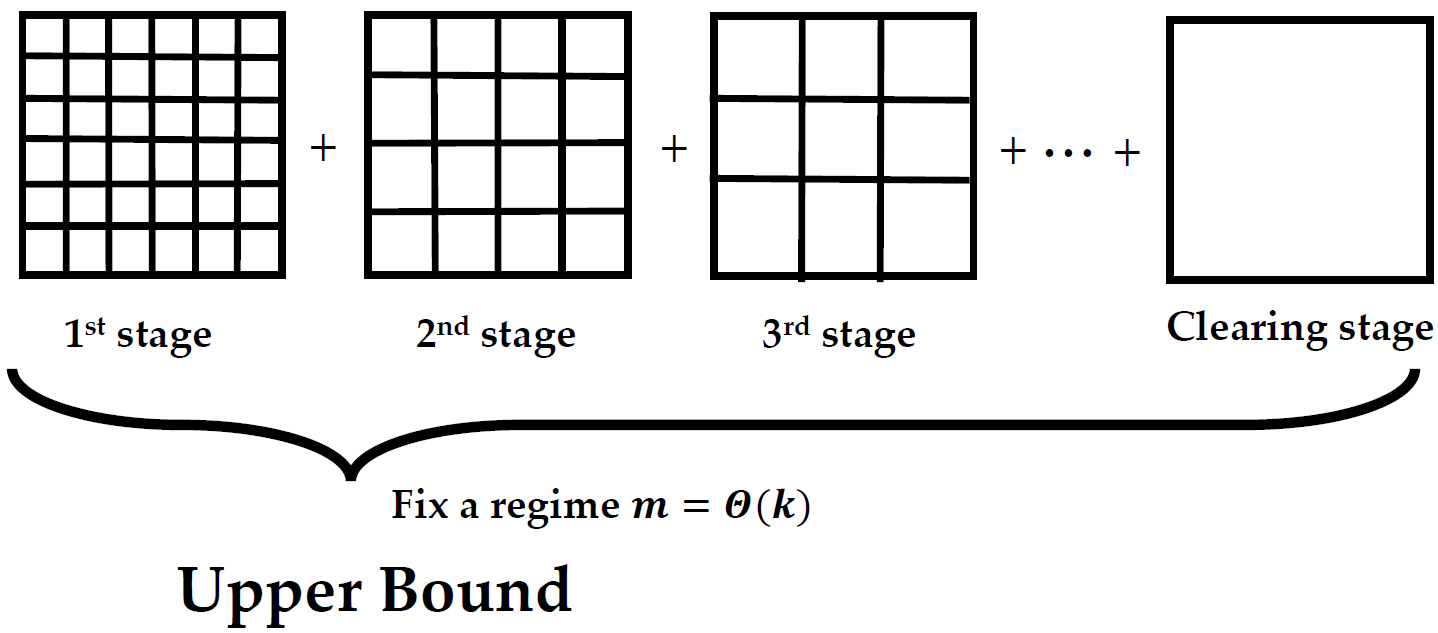}
	\caption{This graph illustrates the flow of distributive-decoding algorithms. We decode local information stage by stage and end up with a clearing stage to improve performance. As a result stated in Theorem~\ref{thm:bitmeter}, under our assumptions and fix our interested regime $m=\Theta(k)$, the \textit{bit-meters} are bounded by $\Theta\left({\sqrt{\li k}}\right)$.}
	\label{fig:dis}
\end{figure}

\section{Algorithms Description}
\label{sec:ad}
 \subsection{\em Chain Algorithm (CA)}
 \label{ad:ca}
 We describe the Chain Algorithm (CA) stage by stage as Figure~\ref{fig:ca} shows. 
 \begin{figure}
 	\centering
 	\includegraphics[scale=0.35]{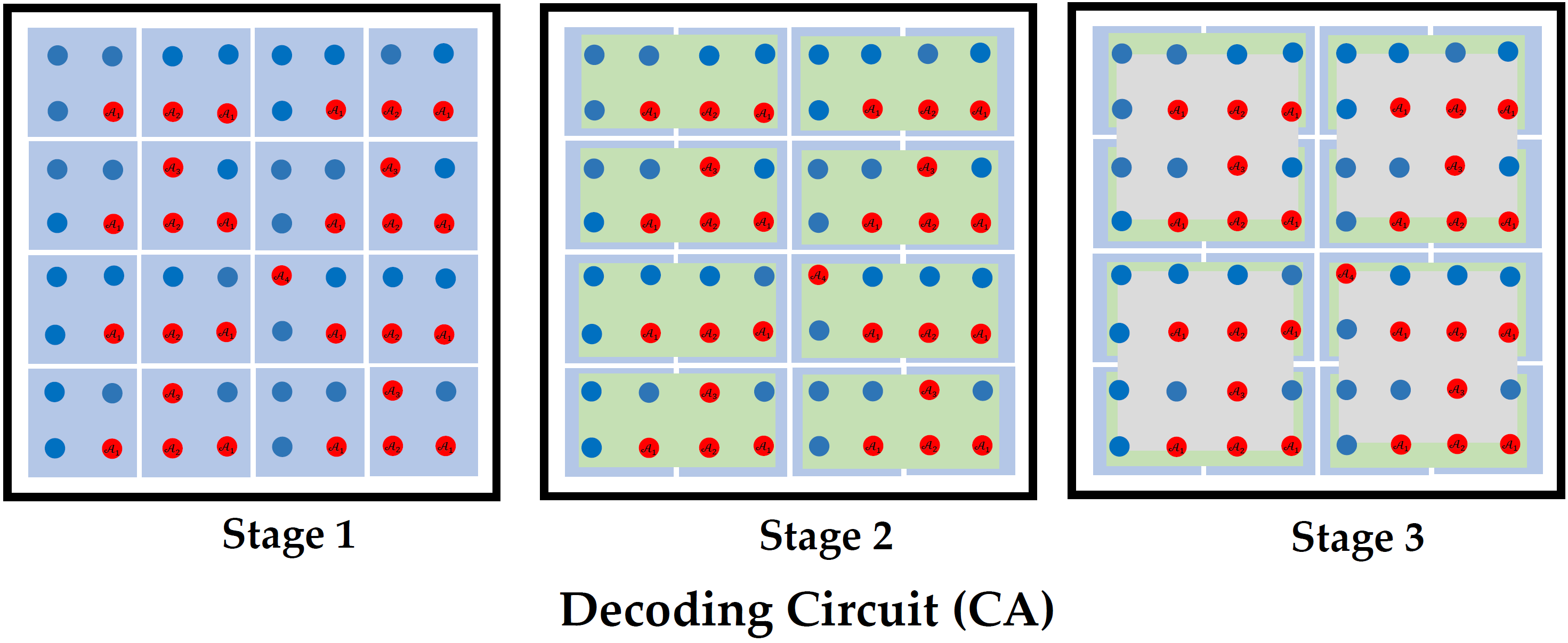}
 	\caption{This figure shows first three stages of Chain Algorithm (CA), with the number of nodes contained in each sub-circuit increase from $4$ to $8$ and $16$.}
 	\label{fig:ca}
 \end{figure}
\label{al:one}
\subsubsection{First Stage}
\label{sec:fs}
%

The {\it{input vector}} of length $\li$ is first divided into $C_{\textnormal{CA}}k$ groups which are compressed separately. $C_{\textnormal{CA}}$ is a constant chosen so as to achieve a desired error probability $\ep$ with details provided in the Appendix~\ref{appendix:up}. Each group contains $\li/C_{\textnormal{CA}}k$ entries. The decoding process is performed independently for each group. Thus, the corresponding decoder only needs to process these groups locally, \textit{i.e.}, it only needs to communicate within the local sub-circuits for the corresponding groups. Intuitively, this method leads to savings in energy since the distances for communication are reduced greatly for most of the communication between nodes. We use a \textit{measurement constant} $c$ number of measurements for each group. For each group, our decoding algorithm aims to resolve the non-zero entry if it contains exactly one non-zero entry. As a result, a constant proportion $\rho$ (depending on $C_{\textnormal{CA}}$ and $c$) of the total $k$ non-zero entries can be located and solved within $\precision$-bits of precision with a high probability. Figure~\ref{fig:st1} illustrates the partition and a possible way to construct the {\em encoding matrix} for finding the single non-zero entry in the {\em input vector}.

\begin{figure}
	\centering
	\includegraphics[scale=0.43]{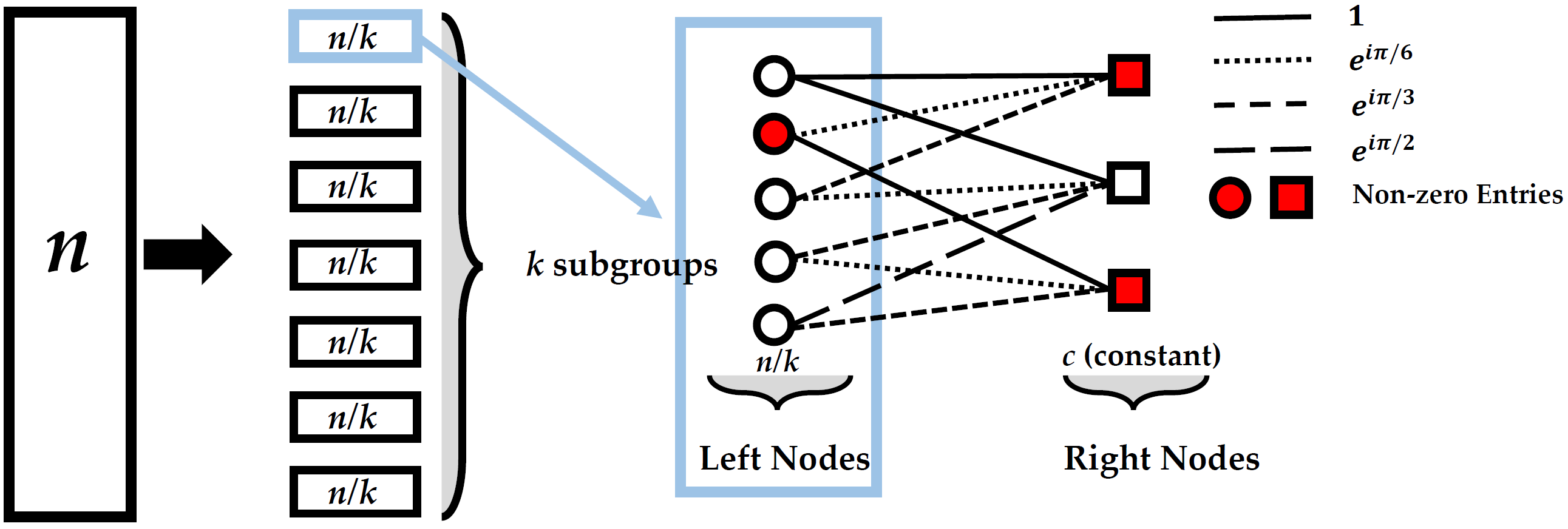}
	\caption{With the assumption $C_{\textnormal{CA}}=1$ and $c=3$ this schematic graph for the \textit{1st} stage illustrates the division of {\it{input vector}} and the construction of {\it{encoding matrix}} $\mathcal{A}$. In the first box with color of calamine blue, we exemplify part of the  \textit{Identification Phase} by a group of five left nodes and three right nodes. Note that for each node on the right, the weights of edges connected to it should be made unique, which is the requirement for \textit{Identification Phase}. In \textit{Verification Phase}, the connection is kept as the same whereas the weights $e^{\iota\Theta_{i,j}^{V}}$ (Here we keep using the same notations in ~\cite{bakshi2012sho}, where $i$ and $j$ are indexes for entries in the encoding matrix $\mathcal{A}$, $\iota$ denotes the positive square root of $-1$ in order to avoid confusions) for non-zero entries in the encoding matrix $\mathcal{A}$ are chosen uniformly in $[0,\pi / 2]$. The constructions are elaborated in section~\ref{mc} for set-up of \textit{encoding matrix}.}
	\label{fig:st1}
\end{figure}

Define $\phi=\lceil 1/(1-\rho)\rceil$ as a parameter for the remaining stages.

\subsubsection{Second Stage up to $\log_\phi(k/\log_2{k})$-th Stage}
\label{sec:sc}

For $i>1$, in the $i$-th stage we combine $\phi=\lceil 1/(1-\rho)\rceil$ of the groups coming from the $\left(i-1\right)$-th stage together. Note that in the $i$-th stage each group is only processed in a local region with area of order approximately $\phi^{i-1}n/C_{\textnormal{CA}}k$. Thus, in the $i$-th stage, $C_{\textnormal{CA}}k/\phi^{i-2}$ groups from the $\left(i-1\right)$-th stage merge into $C_{\textnormal{CA}}k/\phi^{i-1}$ new groups. Each new group contains $\phi^{i-1}\li/C_{\textnormal{CA}}k$ entries of the {\it{input vector}}. Next the decoding algorithm from the first stage~\ref{sec:fs} is implemented on each new group totally, {\em i.e.,} like in the first stage, the corresponding decoders need to handle the information for each group only. This algorithm continues up to $\log_\phi(k/\log_2{k})$ stages using the same \textit{measurement constant} $c$ as the number of measurements for each group. As a result, in the $i$-th stage approximately $\rho$ proportion of the total $C_{\textnormal{CA}}k/\phi^{i-1}$ non-zero entries can be located and solved in $\precision$-bits precision with a high probability which is of order $1-\sqrt{1/k}$. Figure~\ref{fig:st2} illustrates the combination process.

\begin{figure}
	\centering
	\includegraphics[scale=0.34]{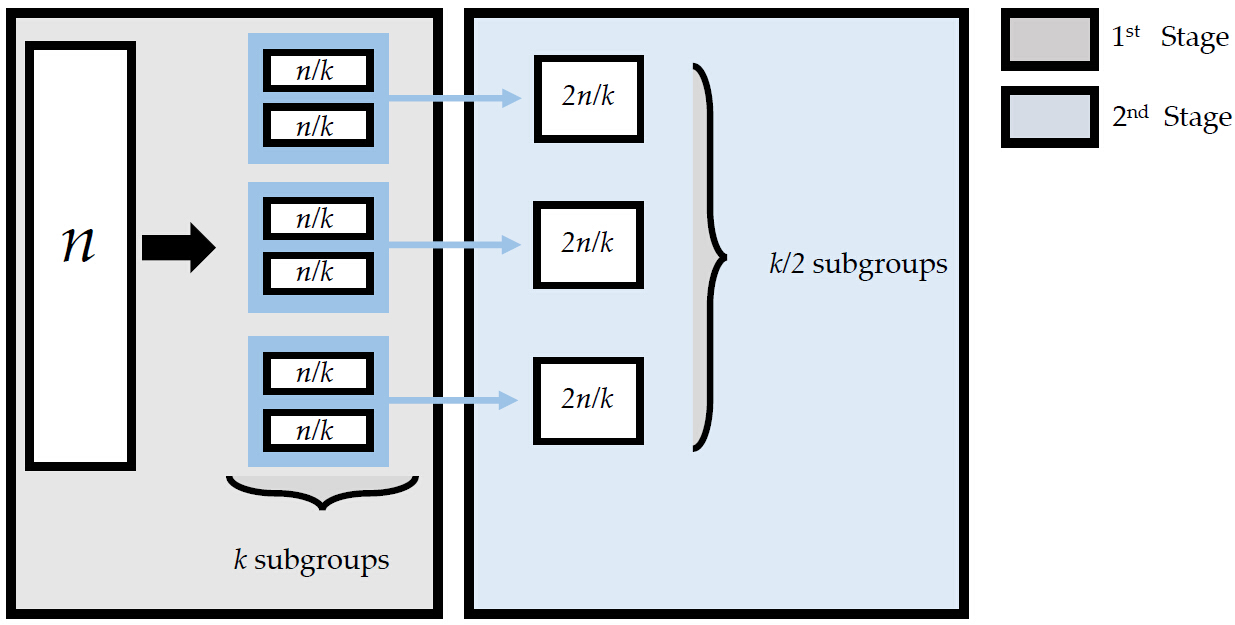}
	\caption{The schematic graph for the \textit{2nd} stage illustrates the combination processes in the coming stages.The groups in the left block with gray color are combined  two-by-two (assume $\rho=1/2$ and set $\phi=2$) to form the new groups in the right block with light blue color. Note that in a same way as the first stage exemplifies in Figure~\ref{fig:st1}, the two \textit{Identification Phase} and \textit{Verification Phase} are also implemented for each of the new group with the same number of measurements $c$ as the previous stages.}
	\label{fig:st2}
\end{figure}

\subsubsection{$\left(\log_\phi(k/\log_2{k})+1\right)$-th Stage (Clearing Stage)}

\begin{figure}
	\centering
	\includegraphics[scale=0.4]{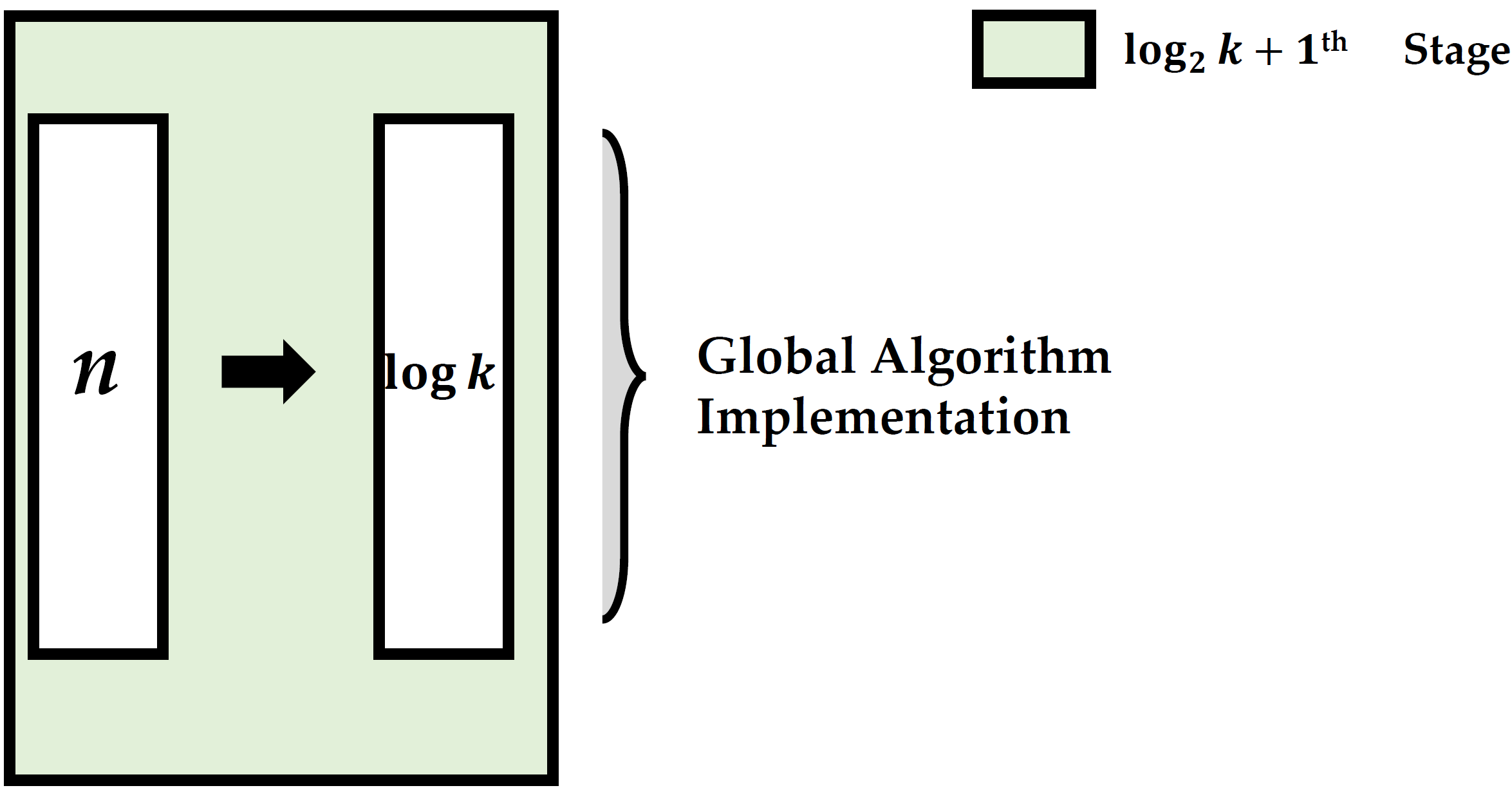}
	\caption{The schematic graph for the last stage illustrates the clearing process. This stage ends up the previous stages using $\Sigma(\log_2 k)$ measurements, and sue global communications between nodes for decoding the entire length-$\li$ {\it{input vector}} such that all the remaining non-zero entries are resolved with a high probability of order $\sqrt{1/k}$.}
   \label{fig:st3}
\end{figure}
After $\log_\phi(k/\log_2{k})$ stages, the algorithm stops forming group of nodes and, instead globally decodes the remaining unsolved non-zeros of {\it{input vector}} $\ivt$ given the information from the previous stages~\ref{sec:fs} and~\ref{sec:sc} along with $\Theta(\sqrt{k})$ new measurements. In contrast to the previous stages, each computed value is potentially communicated across the entire decoding circuit. This helps improve the performance with respect to the error probability. Overall, the algorithm achieves an average block error probability $\ep$ of order $\sqrt{1/k}$ and consumes {\it{bit-meters}} of order $\left(\sqrt{\frac{\li k}{\log{n}}}\sqrt{\log{\frac{1}{\ep}}}\right)$.

 \subsection{\em Shotgun Algorithm (SA)}
  \label{ad:sa}
  Next, we describe the Shotgun Algorithm (SA) stage by stage as Figure~\ref{fig:sa} shows.
\label{al:two}
\subsubsection{First Stage}
\label{sec:fs:sa}

\begin{figure}
	\centering
	\includegraphics[scale=0.35]{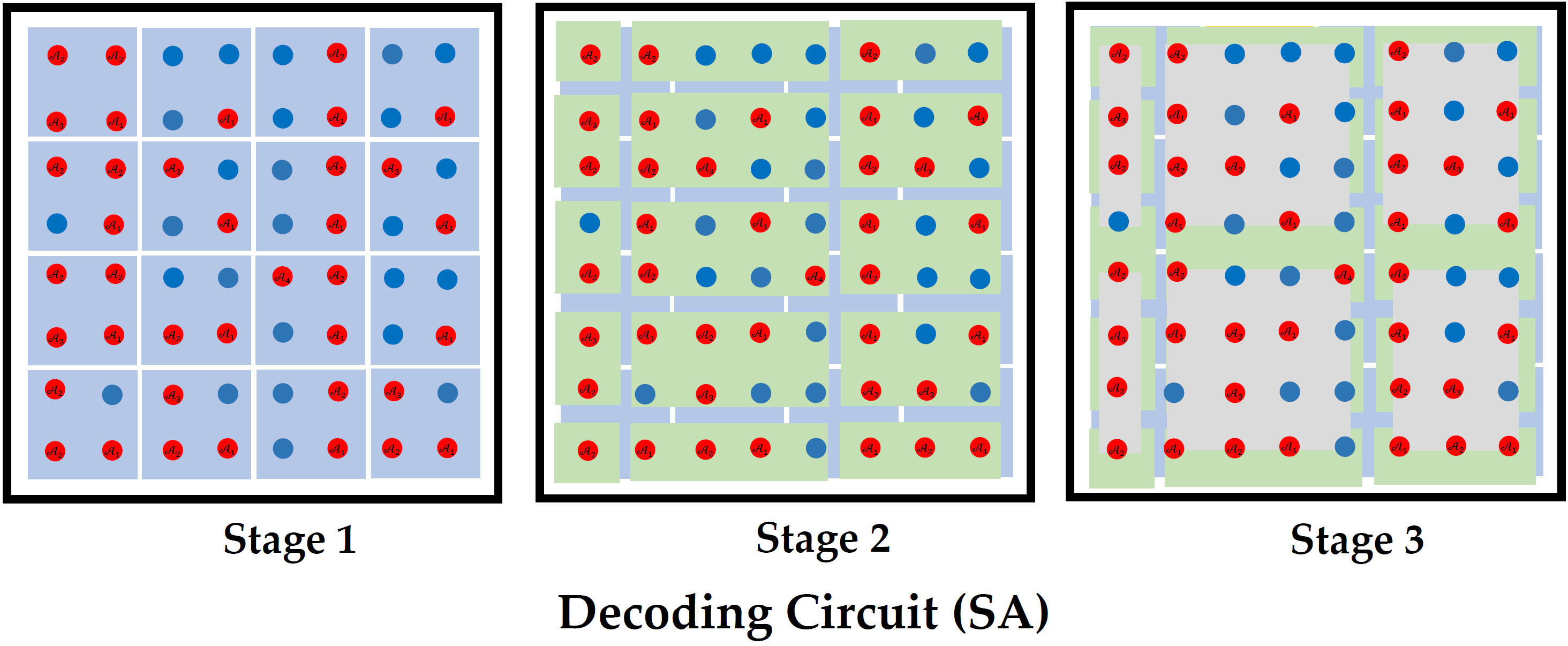}
	\caption{This figure shows first three stages of Shotgun Algorithm (SA), with the number of nodes contained in each sub-circuit increase from $4$ to $8$ and $16$. Note the difference between SA and CA is that the sub-circuits for SA in each stage is chosen uniformly at random instead of by combining previous sub-circuits. }
	\label{fig:sa}
\end{figure}
In a similar way to the Chain Algorithm~\ref{al:one}, the {\it{input vector}} of length $\li$ is first divided into $C_{\textnormal{SA}}k$ groups which are compressed separately. $C_{\textnormal{SA}}$ is a constant for ensuring a desired error probability $\ep$. Each group contains $\li/C_{\textnormal{SA}}k$ entries. The decoding process is performed independently for each group. Thus, the corresponding decoders only need to decode these groups locally, \textit{i.e.}, it only needs to communicate within the local sub-circuits for the corresponding groups. The number of measurements for each group equals the \textit{measurement constant} $c'$. As a result, a constant proportion  $\sigma$ of the total $k$ non-zero entries ({\em i.e.,} a total of $\sigma k$ entries) can be located and solved within $\precision$-bits of precision with a high probability. 

Define $\varphi=\lceil 1/(1-\sigma)\rceil$ as a parameter for the remaining stages.


\subsubsection{Second Stage up to $\log_\varphi(k/\log_2{k})$-th Stage}

In the $i$-th stage we combine $\varphi=\lceil 1/(1-\sigma)\rceil$ of the groups coming from the $\left(i-1\right)$-th stage together by choosing them \textit{uniformly} at random. Note that the combination is performed independently for each stage, and in the $i$-th stage the area spanned by each group is of order approximately $\varphi^{i-1}n/C_{\textnormal{SA}}k$. Thus, in the $i$-th stage, $C_{\textnormal{SA}}k/\varphi^{i-1}$ new groups are formed. Each new group contains $\varphi^{i-1}\li/C_{\textnormal{SA}}k$ entries of the {\it{input vector}}. The decoding algorithm for each group is the same as that of the first stage of the Chain Algorithm (CA) of section~\ref{ad:ca}. Like the first stage of CA, the corresponding decoders need to handle the information locally. The algorithm continues up to the $\log_\varphi(k/\log_2{k})$-th stage. As a result, at the end of $\log_\varphi(k/\log_2{k})$-th stage approximately $\log k$ unsolved non-zero entries remain with a high probability which is of order $1-\sqrt{1/k}$.

\subsubsection{$\left(\log_\varphi(k/\log_2{k})+1\right)$-th Stage (Clearing Stage)}

After $\log_\varphi(k/\log_2{k})$ stages, the algorithm stops combination, and globally decodes the remaining unsolved non-zeros in the {\it{input vector}} $\ivt$ given the information from the previous stages~\ref{sec:fs} and~\ref{sec:sc}. Similar to the last stage of the Chain Algorithm, as all the information is potentially communicated across the entire decoding circuit, the error probability is decreased. In fact, overall, the algorithm achieves an average block error probability $\ep$ of order $\sqrt{1/k}$ while consuming {\it{bit-meters}} of order $\left(\sqrt{\frac{\li k}{\log{n}}}\sqrt{\log{\frac{1}{\ep}}}\right)$.

To summarize, the first Chain Algorithm combines local sub-circuits sequentially and accumulates the information together to resolve the input vector $\ivt$. While the performance of this algorithm is better than the Shotgun Algorithm (SA), a drawback of the Chain Algorithm is that as $\li$ increases, the computation required from the central nodes within each local sub-circuit also increases. In contrast, for SA, except for the clearing stage, every node has the same functionality and the decoder merely needs to decode a possible single non-zero entry in each local sub-circuit. The performance of CA and SA are stated in Theorem~\ref{thm:bitmeter}. Note that it matches the lower bound in Corollary~\ref{corollary} when $\lo=\Theta\left(k\right)$.

\begin{figure*}
	\centering
	\includegraphics[scale=0.5]{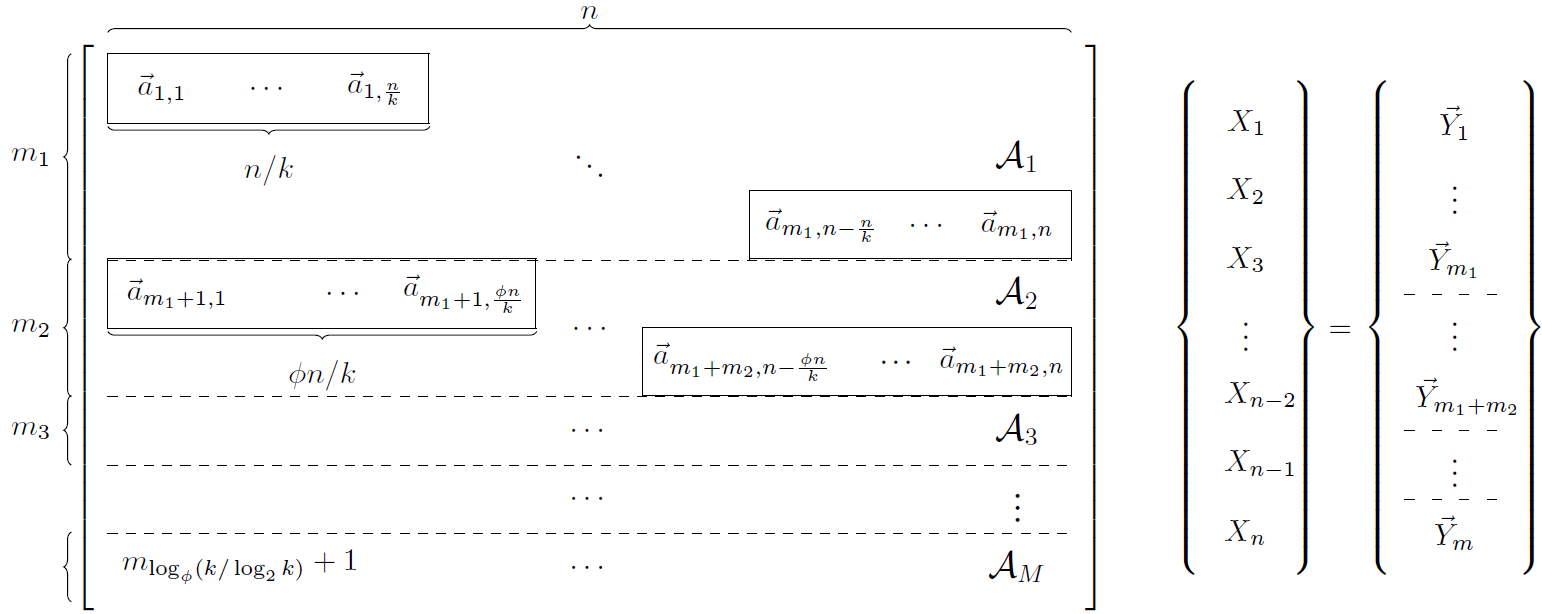}
	\caption{This figure shows in details the construction of \textit{encoding matrix} $\mathcal{A}_{CA}$ where $\phi>1$ is a constant defined at section~\ref{al:one} and the \textit{number of stages} is $M=\lo_{\log_\phi(k/\log_2{k})}+1$. }
	\label{matrix}
\end{figure*}


\subsection{\em Choice of Encoding Matrices}
\label{mc}


For our Chain Algorithm (CA) introduced in Section~\ref{al:one}, the encoding matrix $\mathcal{A}_{\textnormal{CA}}$ is constructed as shown in Figure~\ref{matrix}. Let $\lo_i$ denote the total {\em number of measurements} for the $i$-th stage.

One possible way to construct the entires of the {\em encoding matrix} is by choosing $c=2$ and setting $\vec{a}_{i,j}=0$ if $j$-th item is not inside $i$-th sub-circuit, otherwise $\vec{a}_{i,j}=[e^{\iota\Theta_{i,j}^{I}},e^{\iota\Theta_{i,j}^{V}}]^{T}$ with $\Theta_{i,j}^{I}=\pi ij/\left(2{\li}^2\right)$ and  $\Theta_{i,j}^{V}$ chosen uniformly at random from $[0,\pi/2]$ where $\iota$ denotes the positive square root of $-1$. Therefore, the total number of measurements $m$ follows $m=c\sum_{i=1}^{M}{\lo_i}$ for some \textit{measurement constant}\footnote{The \textit{measurement constant} $c$ can be toned to achieve a desired block error probability $\ep$ by using some additional measurements to verify the linear equations.}. 

For our Shotgun Algorithm (SA) introduced in section~\ref{al:two}, the {\em encoding matrix} $\mathcal{A}_{\textnormal{SA}}$ is generated in a similar manner as above. The only difference between $\mathcal{A}_{\textnormal{CA}}$ and $\mathcal{A}_{\textnormal{SA}}$ is that the sub-matrices $\mathcal{A}_i$ ($i=1,2,\ldots,M$) are no longer related like $\mathcal{A}_{\textnormal{CA}}$ in Figure~\ref{matrix} since the covering sub-circuits are chosen randomly.


\subsection{\em Decoding Steps}

For the Chain Algorithm (CA), the decoding circuit $DecCkt$ stores $m$ received output entries $\vec{Y}_1,\vec{Y}_2,\ldots,\vec{Y}_m$ in $m$ \textit{input-nodes}. Suppose $C_{\textnormal{CA}}=1$, the decoding starts from $\vec{Y}_1$ by checking each group (starting with the group $\{X_1,X_2,\ldots,X_{\li/k}\}$ to determine if the group contains at most one non-zero entry. If so, the \textit{recovery vector} is updated, otherwise, the involved {\em input-nodes} transmit the corresponding entries of \textit{output vector} to the \textit{input-nodes} of a larger sub-circuit containing the current one.  This continues until there is a feasible solution for solving the resulting linear equations. The entire process is written formally as Algorithm~\ref{Algorithm:ca} in the Appendix~\ref{appendix:da}.

%

On the other hand for Shotgun Algorithm (SA), the \textit{input-nodes} need not to pass information to subsequent stages. The decoding circuit $DecCkt$ stores $m$ received output entries $Y_1,Y_2,\ldots,Y_m$ in $m$ \textit{input-nodes}. Similarly to CA, we suppose $C_{\textnormal{SA}}=1$ and the first decoding step starts with $Y_1$ by checking if the group $\{X_1,X_2,\ldots,X_{\li/k}\}$ contains at most one non-zero entry. After that, in the later stages, the size of the sub-circuits increases by a constant $\varphi$ for each stage and \textit{input-nodes} check if the resulting group contains at most one non-zero entry. The entire algorithm is written formally as Algorithm~\ref{Algorithm:sa} shown in Appendix~\ref{appendix:da}. The analysis of CA and SA are provided in Appendix~\ref{appendix:up} and follow the analysis from~\citep{bakshi2012sho}.

\appendices 
\appendices
\section {Proofs of Lower Bound}
\label{appendix:lb}

\begin{figure*}
	\centering
	\includegraphics[scale=0.45]{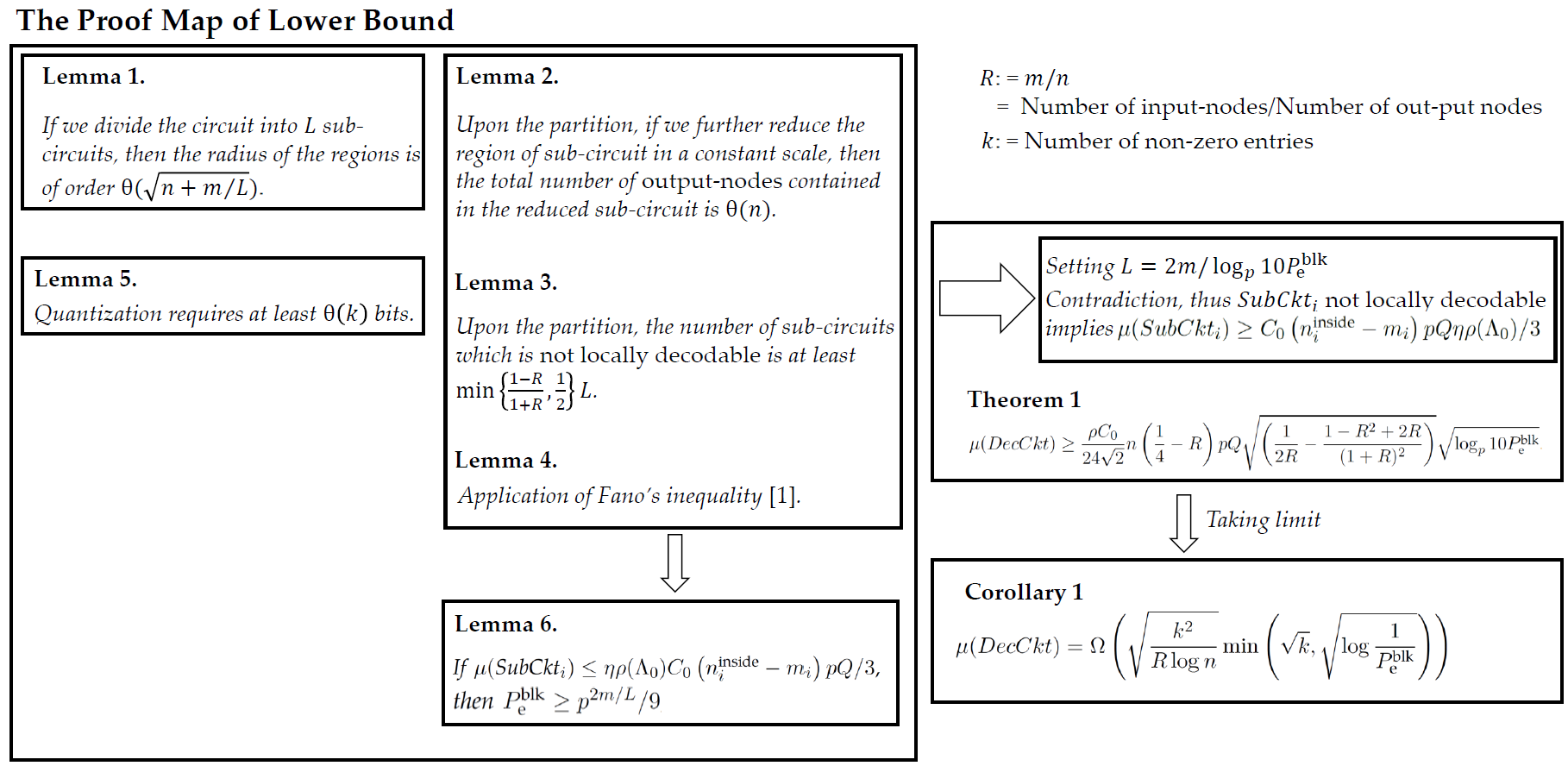}
	\caption{The Proof Map of Lower Bound.}
	\label{fig:pm_ld}
\end{figure*}


Let the {\em packing density} of a lattice $\Lambda$ that spans $\real^2$ be $\sigma(\Lambda)=\frac{\pi {\rho\left(\Lambda\right)}^2}{\det(\Lambda)}$ with $\det(\Lambda)$ denoting the volume of {\em fundamental parallelepiped} of $\Lambda$. Base on Definition~\ref{def:sp}, we derive the following lemma stating the relationship between the {\em paking radius} of $\Lambda$ and $\Lambda_0$.

\begin{lemma}
	\label{lemma_3}
	Let $L$ denote the number of sub-circuits by {\textnormal{Stencil-partition}}. Then the {\textnormal{packing radius}} $\rho(\Lambda_{0})$ of the outer part of sub-circuits is given by
	\begin{align*}
	\rho(\Lambda_{0})=\sqrt{\frac{\sigma{(\Lambda_{0})}(n+m)}{\sigma{(\Lambda)}L}}\rho(\Lambda).
	\end{align*}
\end{lemma}

\begin{proof}
	By Definition~\ref{IM} of the Implementation Model $(\rho(\Lambda),\mu)$, the lattice $\Lambda$ has a packing radius $\rho(\Lambda)>0$, and since it is a 2-D lattice, the \textit{packing density} $\sigma(\Lambda)$ is given by \[\sigma(\Lambda)=\frac{\pi\rho^{2}(\Lambda)}{\det(\Lambda)}.\]
	
	Similarly for the sub-lattice $\Lambda_{0}$, we also have a positive \textit{packing density} $\sigma(\Lambda_{0})>0$ such that $\sigma(\Lambda_{0})=\pi\rho^{2}(\Lambda_{0})/{\det(\Lambda_{0})}$. Moreover, the cardinality of the quotient $\Lambda/\Lambda_{0}$ equals to $\det(\Lambda_{0})/\det(\Lambda)$. Since $L=\frac{n+m}{|\Lambda/\Lambda_{0}|}$, we conclude that $\rho(\Lambda_{0})=\sqrt{\sigma{(\Lambda_{0})(n+m)}/{\sigma{(\Lambda)}L}}\rho(\Lambda)$.
	
\end{proof}

\begin{lemma}
	\label{lemma_1}
	Consider the {\textnormal{Implementation Model $(\rho,\mu)$}}. For any {\em fractional parameter} $\eta>0$, there exists a point $\textbf{u}\in\Lambda$ for {\textnormal{Stencil $(\lambda,\eta,\textbf{u})$}} such that the number of {\em output-nodes} covered by the Stencil is bounded from below by
	\begin{align}\label{sum_of_innodes}
	\sum_{i}^{L}{\li_i^{\textnormal{inside}}} &\geq \li\left(1-2\eta\right)^2.
	\end{align}
\end{lemma}

\begin{proof}
	Note that $\li\left(1-2\eta\right)^2$ is the expected number of \textit{output-nodes} covered by the Stencil, if the point $\textbf{u}\in \Lambda$ is uniformly distributed. Thus there exists at least one point $\textbf{u}$ that satisfies the bound in~\eqref{sum_of_innodes}.
\end{proof}

\begin{lemma}
	\label{lemma_2}
	Consider the {\textnormal{Implementation Model $(\rho,\mu)$}}. Let $L$ be the number of sub-circuits. For any {\em fractional parameter} $\eta>0$ and any choice of $\textbf{u}\in\Lambda$ for {\textnormal{Stencil $(\lambda,\eta,\textbf{u})$}}, the number of sub-circuits satisfying $m_i\leq\min\{2m/{L},n_i\}$ is larger or equal to $\min\{\left(1-R\right)/\left(1+R\right),1/2\}L$ where $R$ is the {\textnormal{rate}} of compressive sensing defined by $R=m/n$.
\end{lemma}

\begin{proof}
	Assume $n>m$, first we choose the point $\textbf{u}\in \Lambda$ of the Stencil such that the location of {\textnormal{output-nodes}} satisfies~\eqref{sum_of_innodes}. Then, for this fixed choice of $\textbf{u}\in \Lambda$, we consider the worst location of the {\textnormal{input-nodes}} which minimizes the fraction of {\textnormal{input-nodes}} satisfying $m_i\leq\min\left(2m/{L},n_i\right)$. We call a sub-circuit \textit{non-locally decodable} if $m_i\leq\min\left(2m/{L},n_i\right)$ and \textit{locally decodable} otherwise\footnote{Actually $m_i\leq\min\left(2m/{L},n_i\right)$ is merely a sufficient condition for a sub-circuit to be non-locally decodable, however, we will use the term \lq\lq{}non-locally decodable\rq\rq{} to imply that the sub-circuit satisfies $m_i\leq\min\left(2m/{L},n_i\right)$.}. Figure \ref{fig:lemma} gives an example of \textit{locally decodable} and \textit{non-locally decodable} sub-circuits. First, we note that the fraction of sub-circuits satisfying $m_i\leq 2m/L$ is $1/2$. Similarly, the fraction of sub-circuits satisfying $m_i\leq n_i$ is at least $\left(1-R\right)/\left(1+R\right)$. Thus, the fraction of sub-circuits satisfying $m_i\leq\min\left(2m/{L},n_i\right)$ is at least $\min\{\left(1-R\right)/\left(1+R\right),1/2\}$. Below we prove the claim explicitly.
	
	Note that the number of nodes in each sub-circuit is $\li_i+\lo_i=\frac{\li+\lo}{L}$. Let $\alpha=\min\{\frac{2\lo}{L},\frac{\li+\lo}{2L}\}$. Since $\lo_i\leq \alpha$ for each {\em non-locally decodable} sub-circuit, it satisfies $\frac{\lo_i}{\lo_i+\li_i}\leq \frac{2\lo}{\lo+\li}$ and $\frac{\lo_i}{\lo_i+\li_i}\leq \frac{1}{2}$. 
	
	Now we consider two cases independent with the sub-circuits:
	\begin{enumerate}
		\item $\lo<\frac{\li}{3}$ $\Rightarrow$ A $SubCkt$ is {\em non-locally decodable} if $\frac{\lo_i}{\lo_i+\li_i}\leq \frac{2\lo}{\lo+\li}$;
		\item $\lo\geq \frac{\li}{3}$ $\Rightarrow$ A $SubCkt$ is {\em non-locally decodable} if $\frac{\lo_i}{\lo_i+\li_i}\leq \frac{1}{2}$.
	\end{enumerate}
	
	Hence the fraction $f_{\textnormal{NLD}}$ of {\em non-locally decodable} sub-circuits satisfies $f_{\textnormal{NLD}}\geq \min(1/2,(1-R)/(1+R))$.
\end{proof}

Next we state a lemma derived from Fano's inequality~\citep{cover2012elements}.

\begin{lemma}
	\label{lemma_4}
	If at most $\ent\left(\ivt\right)/3$ bits of information are available to obtain an estimate $\rvt$ of a variable $\ivt$ with entropy $\ent\left(\ivt\right)$, then $\Pr\left[\rvt \neq\ivt\right]\geq 1/9$.
\end{lemma}

\begin{proof}
	Similar to the proof of Fano's inequality~\citep{cover2012elements}, we define the error random variable $E(\ivt,\rvt)$ as follows:
	\begin{align*}
	E(\ivt,\rvt) = \begin{cases}
	1  \text{ if $\rvt=\ivt$}\\
	0  \text{ if $\ivt\neq\ivt$}.
	\end{cases}
	\end{align*}
	
	Since the input vector $\ivt$, the output vector $\ovt$ and the recovery vector $\rvt$ form a Markov chain $\ivt\rightarrow\ovt\rightarrow\rvt$, we get
	$\ent\left(\ivt\right) = \ent\left(\ivt|\rvt\right)+\mif\left(\ivt;\rvt\right)\leq \ent\left(\ivt|\rvt\right)+\mif\left(\ivt;\ovt\right)\leq\ent\left(\ivt|\rvt\right)+\ent\left(\ovt\right)$.
	
	Thus,
	\begin{align*}
	&\ent\left(\ivt|\rvt\right) \\
	&= \ent\left(E(\ivt,\rvt),\ivt|\rvt\right)\\
	&= \ent\left(E(\ivt,\rvt)|\rvt\right) \\
	&+ \Pr\left[E(\ivt,\rvt)=0\right]\ent\left(\ivt|\rvt,E(\ivt,\rvt)=0\right)\\
	&+ \Pr\left[E(\ivt,\rvt)=1\right]\ent\left(\ivt|\rvt,E(\ivt,\rvt)=1\right)\\
	&\leq h_b\left(P_e\right)+P_e\ent\left(\ivt\right).
	\end{align*}
	
	Given the available information $\mif$ of at most $\ent\left(\ivt\right)/3$ bits, the error probability $P_e:=\Pr\left(\rvt \neq\ivt\right)$ is lower bounded by
	\begin{align*}
	P_e\ent\left(\ivt\right)+h_b\left(P_e\right) 
	&\geq \ent\left(\ivt\right)- \ent\left(\ovt\right)\\
	&\geq \ent\left(\ivt\right)-\ent\left(\ivt\right)/3 \\
	&= 2\ent\left(\ivt\right)/3,
	\end{align*}
	where $h_b\left(\cdot\right)$ on the LHS is the binary entropy function (will also appear in the later parts). 
	
	Then since $n>m>1$ we have 
	\begin{align*}
	P_e &\geq  \frac{2\ent\left(\ivt\right)/3-1}{\ent\left(\ivt\right)} \geq \frac{2}{3}-\frac{1}{2}>\frac{1}{9}.
	\end{align*}
\end{proof}

\begin{lemma}
	\label{lemma_20}
	Consider the {\textnormal{Sparsity Model $(n,m,p)$}}. For every decoding circuit $DecCkt$  on the {\textnormal{Implementation Model $(\rho,\mu)$}}, if the {\textnormal{relative error}} satisfies $||\ivt-\rvt||_{\nr}/||\ivt||_{\nr}\leq 2^{-\precision}$, then there exists a constant $C_0=C_{0}(\ivt,q)<1$ such that asymptotically at least $C_{0}\li p\precision$ bits are required by all the {\textnormal{output-nodes}}.
\end{lemma}

\begin{proof}
	For each $i$, let $\precision_{i}$ denote the number of bits of quantization required to distinguish $\hat{X}_{i}$ and $X_i$ for each entry $\hat{X}_{i}$. Thus, we have $2^{-\precision_{i}-1}\leq\frac{|X_{i}-\hat{X}_{i}|}{|X_{i}|}\leq2^{-\precision_{i}+1}$ for all $i$. Let $||\ivt-\rvt||_{\nr}/||\ivt||_{\nr}\leq 2^{-\precision}$. Hence,
	\begin{align*}
	{\left(\sum_{i=1}^{k}|X_{i}-\hat{X}_{i}|^{q}\right)}^{1/q}\leq 2^{-\precision}{\left(\sum_{i=1}^{k}|X_{i}|^{q}\right)}^{1/q},
	\end{align*}
	which implies that
	\begin{align*}
	{\left(\frac{2^{-q(\precision_{i}+1)}\sum_{i=1}^{k}|X_{i}|^{q}}{\sum_{i=1}^{k}|X_{i}|^{q}}\right)}^{1/q}\leq 2^{-\precision}.
	\end{align*}
	
	By assumption, $0\leq|X_i|\leq \up$ for each $i$ for some constant $\up\geq 0$. By Jensen's inequality (see, for instance in the book~\cite{kuczma2008introduction}), we get 
	\begin{align*}
	|\up|^{q}\sum_{i}^{k}\precision_{i}\geq\sum_{i=1}^{k}|X_{i}|^{q}\precision_{i}\geq \sum_{i=1}^{k}|X_{i}|^{q}\precision.
	\end{align*}
	
	Thus $\sum_{i}^{k}\precision_{i}\geq C_{0}kQ$ with $C_{0}=\sum_{i=1}^{k}|X_i|^{q}/k|\up|^{q}<1$. The asymptotic result follows as $n\rightarrow \infty$.
\end{proof}
\begin{figure*}
	\centering
	\includegraphics[scale=0.43]{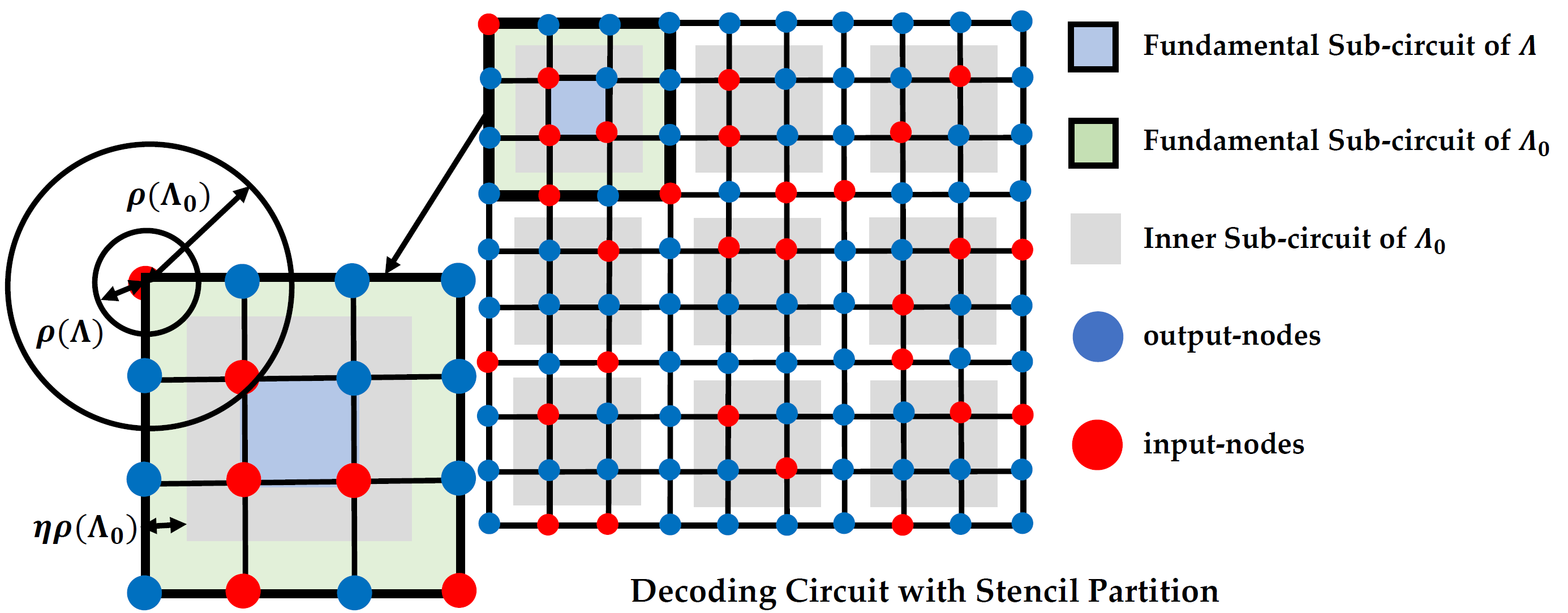}
	\caption{This graph illustrates the \textit{stencil-partition} on the decoding circuit. The sub-lattice which has a larger {\em fundamental parallelepiped} (sub-circuit) defines the sub-circuits. And the {\em inner part} of sub-circuits are fixed by choosing a \textit{fractional parameter} $0<\eta<1$. For instance, for the sub-circuit on the left-up corner contains the {\em fundamental parallelepiped}, the order of quotient $\lambda=|\Lambda/\Lambda_0|=9$, and it has $6$ \textit{input-nodes} and $10$ \textit{output-nodes}. Moreover, it has $1$ \textit{output-node} in the \textit{inner part} of sub-circuit.}
	\label{fig:stencil}
\end{figure*}

Next we combine the lemmas above to give a result connecting {\em bit-meters} and {\em average block error probability} $\ep$. As mentioned before, in this lemma we call the {\em inner part} of a sub-circuit the {\em {inner parallelepipeds}} and the {\em outer part} the \textit{outer parallelepipeds} respectively. 

\begin{lemma}
	\label{lemma_5}
	Consider the {\textnormal{Sparsity Model $(n,m,p)$}}. Let $SubCkt_i$ be a sub-circuits with $m_i\leq\min\left(2m/{L},n_i\right)$ that is obtained via {\em stencil-partitioning} a decoder circuit $DecCkt$ implemented on the {\textnormal{Implementation Model $(\rho,\mu)$}}. If $\mu(\textsl{SubCkt}_i)\leq\eta \rho(\Lambda_{0})C_{0}\left(n_i^{\textnormal{inside}}-m_i\right)p\precision/3$, then $\ep\geq p^{2m/L}/9$, where $L$ is the number of sub-circuits.
\end{lemma}


\begin{proof}
	In each $i$-th sub-circuit $SubCkt_i$, if $m_i\leq\min\left(2m/{L},n_i\right)$, then the number of {\it{bit-meters}} for $SubCkt_i$ is smaller than $\eta \rho(\Lambda_{0})C_{0}\left(n_i^{\textnormal{inside}}-m_i\right)p\precision/3$. Further, the distance between the {\it{outer parallelepipeds}} and the {\it{inner parallelepipeds}} is bounded from below by $\eta \rho(\Lambda_{0})$. Therefore at most $C_{0}\left(n_i^{\textnormal{inside}}-m_i\right)p\precision/3$ bits of information $\mif$ can be communicated from outside the \textit{outer parallelepipeds} to the inside of \textit{inner parallelepipeds}.
	
	Now since $m_i\leq n_i$, if the $n_i$ {\it{output-nodes}} correspond to more than $m_i$ non-zero entries in the {\it{input vector}}, then the decoder cannot determine all $\precision$ bits in {\it{output-nodes}}. We denote this failure event by $\mathcal{L}$. Then $\mathcal{L}$ occurs with probability at least $p^{2m/L}$ since $m_i \leq 2m/L$.
	
	Conditioning on the event $\mathcal{L}$, applying Lemmas~\ref{lemma_4} and \ref{lemma_20} using Fano's inequality~\citep{cover2012elements}, as the received entropy is smaller than $C_{0}\left(n_i^{\textnormal{inside}}-m_i\right)p\precision/3$, the average block error probability is larger than $1/9$. Thus, given the assumptions of this lemma, the (unconditional) error probability for recovering the $n_i$ entries of input vector $\ivt$ with precision $\precision$ in the $i$-th sub-circuit is lower bounded by $p^{2m/L}/9$. Since the average block error probability $\ep$ for the entire circuit is larger than that for any sub-circuit, the claimed result follows.
\end{proof}
\vspace{10pt}


\subsection{\em Proof of Theorem~\ref{thm:lower bound}}

	The \textit{outer parallelepipeds} (or we call it sometimes \textit{outer part} of sub-circuit) of the Stencil divide the circuit into $L$ sub-circuits. Let the $i$-th sub-circuit have $m_i$ {\it{input-nodes}} and $n_i$ {\it{output-nodes}} within the {\it{outer parallelepipeds}} and $n_i^{\textnormal{inside}}$ {\it{output-nodes}} inside the {\it{inner parallelepipeds}}. Using Lemma~\ref{lemma_1} and Lemma~\ref{lemma_2} we can choose a fixed origin $O$ of the Stencil such that at least $(1-2\eta)^2$ fraction of the $\li$ {\it{output-nodes}} are covered by the {\it{inner parallelepipeds}}. Moreover, note that the number of sub-circuits covered by the {\it{inner parallelepipeds}} with $m_i\leq\min\{2m/{L},n_i\}$ is at least $\min\{(1-R)/(1+R),1/2\}L$, which will be used in the later part.
	
	Next, setting $L=2m/\log_p{10\ep}$ in Lemma~\ref{lemma_5}, if we assume that the {\it{bit-meters}} used by a {\em non-locally decodable} sub-circuit is smaller than $\eta \rho(\Lambda_{0})C_{0}\left(n_i^{\textnormal{inside}}-m_i\right)p\precision/3$, then the average block error probability $\ep$ is bounded from below as
	\begin{align*}
	\ep \geq p^{\frac{2m}{L}}/9 = p^{\log_p{10\ep}}/9 = 10\ep/9.
	\end{align*}
	
	Since the above is a contradiction, for each {\em non-locally decodable} sub-circuit $SubCkt_i$, denote
	$\mu\left(i\right)=
	\mu(SubCkt_i) \geq C_{0}\left(n_i^{\textnormal{inside}}-m_i\right)p\precision\eta \rho(\Lambda_{0})/3.
	$
	
	We bound the total {\it{bit-meters}} in the decoding circuit by
	\begin{align}
	\IEEEnonumber
	\mu{(DecCkt)} 
	&\geq \sum_{i=1}^{L}\mu(i)\\ 
	\IEEEnonumber
	&\geq \sum_{m_i\leq\min\{2m/{L},n_i\}}\mu(i)+\sum_{m_i\geq 2m/{L}}\mu(i)\\ 
	\IEEEnonumber
	&\geq \sum_{m_i\leq\min\{2m/{L},n_i\}} C_{0}\left(n_i^{\textnormal{inside}}-m_i\right)
	p\precision\eta \rho(\Lambda_{0})/3 +\sum_{m_i\geq 2m/{L}}\mu(i).
	\end{align}
	
	Now we define three types of sub-circuits under the condition $m_i\geq 2m/{L}$ and $m_i\leq\min\{2m/{L},n_i\}$. First we use ${\textnormal{LD1}}$ to denote those values of $i$ such that $m_i\geq 2m/{L}$ and $\mu(i)\geq {C_{0}\left(n_i^{\textnormal{inside}}-m_i\right)p\precision\eta \rho(\Lambda_{0})/3}$. Next let ${\textnormal{LD2}}$ denote those values of $i$ such that $m_i\geq 2m/{L}$ and $\mu(i)< {C_{0}\left(n_i^{\textnormal{inside}}-m_i\right)p\precision\eta \rho(\Lambda_{0})/3}$. Finally, let ${\textnormal{NLD}}$ denote those values of $i$ such that $m_i\leq\min\{2m/{L},n_i\}$, then it follows that 
	\begin{align}
	\IEEEnonumber
	\mu{(DecCkt)} 
	&\geq \sum_{i\in {\textnormal{LD1}}\cup {\textnormal{NLD}}} {C_{0}\left(n_i^{\textnormal{inside}}-m_i\right)p\precision\eta \rho(\Lambda_{0})/3} 
	+\sum_{i\in {\textnormal{LD2}}}\mu(i)\\
	\IEEEnonumber
	&\myeqa  \frac{1-R}{2(1+R)}\sum_{i=1}^{L} C_{0}\left(n_i^{\textnormal{inside}}-m_i\right)p\precision\eta \rho(\Lambda_{0})/3.\\
	\label{bit-meters}
	&\myeqb \frac{1-R}{2(1+R)}C_{0}\left((1-2\eta)^2n-m\right)p\precision\eta \rho(\Lambda_{0})/3.
	\end{align}
	
	In the above, (a) follows from Lemma~\ref{lemma_2} that the fraction of sub-circuits $SubCkt_i$ with $i\in LD_3$ is larger than $\min\{(1-R)/(1+R),1/2\}$ hence $(1-R)/2(1+R)$ and (b) follows from Lemma~\ref{lemma_1} such that $\sum_{i}^{L}{n_i^{\textnormal{inside}}} \geq n\left(1-2\eta\right)^2$. 
	
	Next, by Lemma~\ref{lemma_3}, 
	
	\begin{align*}
	\rho(\Lambda_{0}) &\leq \frac{1}{2}\rho(\Lambda)\sqrt{ \left(n+m\right)/L }
	=\frac{1}{2}\rho(\Lambda)\sqrt{ \log_p{10\ep}\left(1+\frac{1}{2R}\right)}.
	\end{align*}
	
	 Substituting $\rho(\Lambda_{0})$ into~\eqref{bit-meters}, we get
	\begin{align*}
	\mu{(DecCkt)} 
	&\geq \frac{\eta\rho(\Lambda)}{6\sqrt{2}}C_{0}\left((1-2\eta)^2n-m\right)
	p\precision\sqrt{\left(\frac{1}{2R}-\frac{1-R^2+2R}{(1+R)^2}\right)}\sqrt{\log_p{10\ep}}.
	\end{align*}
	
	
	Choosing $\eta=1/4$ yields Theorem~\ref{thm:lower bound}.
\qed
\vspace{10pt}

For the regime $\lo=\Theta(k)$, we derive the following order expression. This serves as a benchmark for design of our algorithms.




\subsection{\em Proof of Corollary~\ref{corollary}}

	In the {\it{Sparsity Model ($\li$,$m$,$p$) }}, the expected number of non-zero entries in the input vector $\ivt$ is $k=np$. By Hoeffding's inequality, we can bound the number of non-zero entries in the input vector $\ivt$ in the {\it{Sparsity Model ($\li$,$m$,$p$) }} in the range $[k/2,3k/2]$ with probability $1-e^{-\Theta(k)}$. Hence {\it{asymptotically}} we can substitute $k=np$ in the inequality~\ref{lower bound} and get
	\begin{align}
    \label{coro}
	\mu{(DecCkt)} 
	&\geq\frac{\rho(\Lambda)}{24\sqrt{2}}C_{0}\left(\frac{n}{4}-m\right)\frac{k\precision}{n}
	\sqrt{\left(\frac{1}{2R}-\frac{1-R^2+2R}{(1+R)^2}\right)}\sqrt{\log_p{10\ep}}
	\end{align}
	which differs from the original lower bound in the inequality~\ref{lower bound} by a constant.
	
	
	Since $k=o(n)$ by our {\it{sparse assumption}}, in the regime $m=\Theta(k)$, we get $R=m/n=o(1)$. Finally, letting $R=m/n$ we can asymptotically bound the {\it{bit-meters}} as
	\begin{align*}
	&\mu{(DecCkt)} = \Omega\left( \sqrt{\frac{\li k^2}{\log \li}}\min\left(\sqrt{\frac{k}{m}},\sqrt{\frac{\log{\frac{1}{\ep}}}{m}}\right)\right).
	\end{align*}
\qed

\section{Decoding Algorithms} 
\label{appendix:da}
We give the following algorithms descriptions for CA and SA.

\begin{algorithm}[H]
	\caption{Decoding Algorithm (CA)}
	\label{Algorithm:ca}
	\footnotesize
	\begin{algorithmic}[1]
		\Procedure{Dec}{$\ovt,\mathcal{A}_{\textnormal{CA}}$}
		\For{$i \leftarrow 1,\lo $}
		\State load $c'$ rows of \textit{encoding matrix} $\mathcal{A}_{\textnormal{CA}}^{i}$ and $\vec{Y}_i$ at the $i$-th \textit{input-node}
		\State load $\mathcal{S}_i=\{\vec{Y}_t\}_{t\in\{1,2,\ldots,i-1\}}$ received from previous \textit{input-nodes} and their corresponding rows of \textit{encoding matrix} $\mathcal{A}_{\textnormal{CA}}^{t}$
		\State $flag=0$
		\For{$j\leftarrow 1$ to $\li$ }
		\If{$\exists$ a real number $b$ such that $b\vec{a}_{i,j}=\vec{Y}_i$}
		\State $X_j=b$ 
		\State $flag=1$
		\State send $b$ to the $j$-th \textit{output-node}
		\State update the \textit{encoding matrix} $\mathcal{A}_{\textnormal{CA}}$
		\State break
		\Else 
		\State continue    
		\EndIf
		\EndFor
		\If {$flag=0$}
		\If {$\exists$ $c\left(|\mathcal{S}_i|+1\right)$ where $c$ is the \textit{measurement constant} non-zero real numbers as non-zero entries of the updated \textit{recovery vector} $\rvt$ such that the linear equations $\{\vec{Y}_t=\mathcal{A}_{\textnormal{CA}}^{t}\rvt\}_{t\in\{1,2,\ldots,i-1\}}$ hold}
		\State update the \textit{recovery vector} $\rvt$
		\State update the \textit{encoding matrix} $\mathcal{A}_{\textnormal{CA}}$
		\Else
		\State send $\vec{Y_i}$ to the \textit{input-node} corresponding to the sub-circuit in the $i+1$-th stage covering the current one
		\EndIf
		\EndIf
		\EndFor
		\State clearing stage
		\EndProcedure
	\end{algorithmic}
\end{algorithm}

\begin{algorithm}[H]
		\footnotesize
	\caption{Decoding Algorithm (SA)}
	\label{Algorithm:sa}
	\begin{algorithmic}[1]
		\Procedure{Dec}{$\ovt,\mathcal{A}_{\textnormal{SA}}$}
		\For{$i \leftarrow 1,\lo $}
		\State load $c'$ rows of \textit{encoding matrix} $\mathcal{A}_{\textnormal{SA}}^{i}$ and $\vec{Y}_i$ at the $i$-th \textit{input-node}
		\For{$j\leftarrow 1$ to $\li$ }
		\If{$\exists$ a real number $b$ such that $b\vec{a}_{i,j}=\vec{Y}_i$}
		\State $X_j=b$ 
		\State send $b$ to the $j$-th \textit{output-node}
		\State update the \textit{encoding matrix} $\mathcal{A}_{\textnormal{SA}}$
		\State break
		\Else 
		\State continue
		\EndIf
		\EndFor
		\EndFor
		\State clearing stage
		\EndProcedure
	\end{algorithmic}
\end{algorithm}

\section {Proofs of Upper Bounds}
\label{appendix:up}

The outer bound is achieved by performing measurements according to a specially designed $m\times\li$ complex matrix $\mathcal{A}_{m,\li}$, and then The perform decoding in a stage-by-stage manner. First, we state a lemma describing some geometric properties that follow from our definitions of models and descriptions of algorithms.

\begin{figure*}
	\centering
	\includegraphics[scale=0.48]{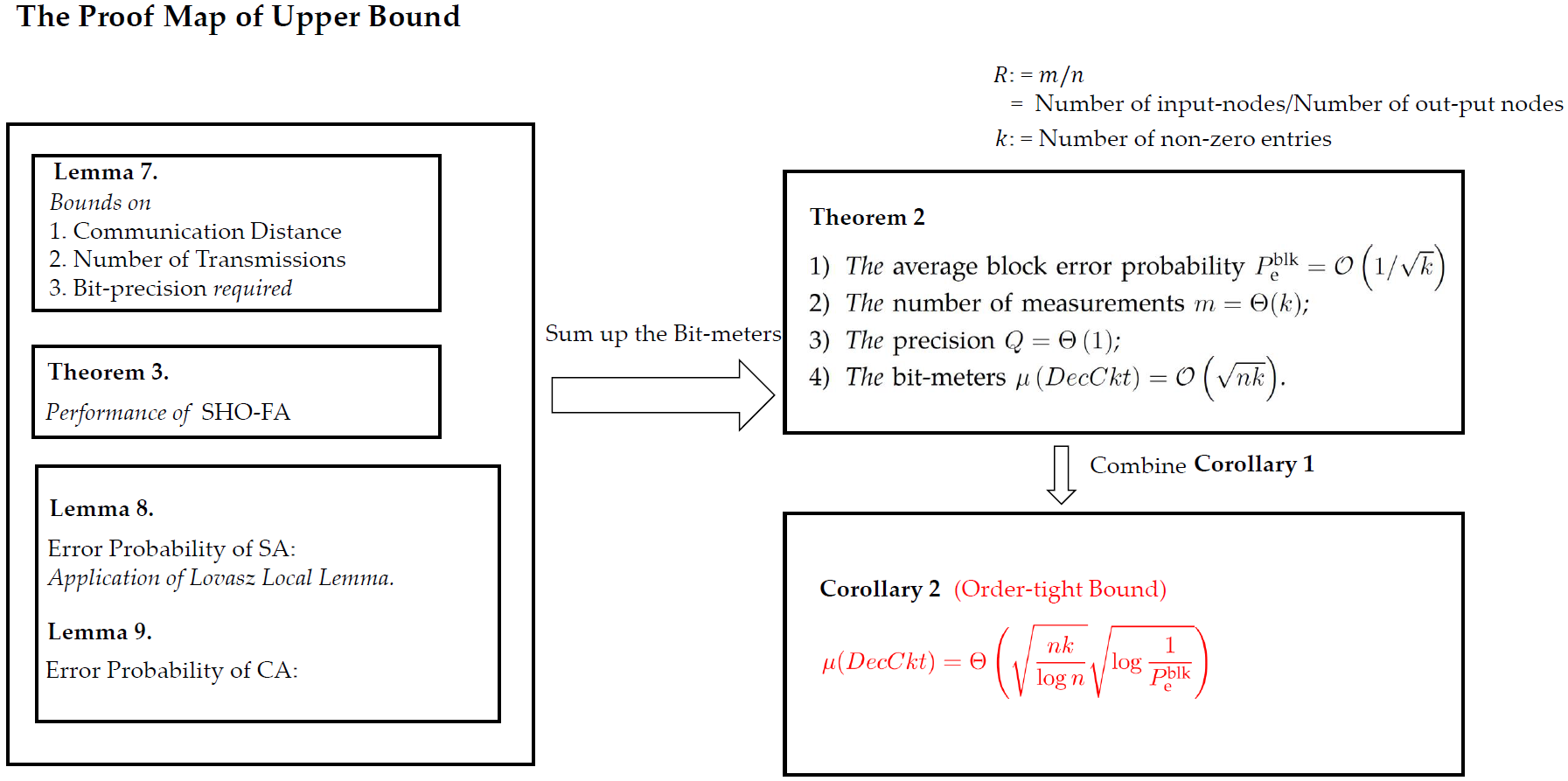}
	\caption{The Proof Map of Upper Bound.}
	\label{fig:pm_up}
\end{figure*}
\begin{lemma}[Properties of $DecCkt$]
	\label{lemma:prop}
	A {\em decoding circuit} $DecCkt$ implementing the decoding steps defined by {\em CA} and {\em SA}, it satisfies the following properties (here $i$ denotes the {\em index of stages}): 
	
	\begin{itemize}
		\item The {\textnormal{Communication Distance}} $\mathcal{D}_{DecCkt}$ is bounded from above by
		
		\begin{align*}
		\mathcal{D}_{DecCkt}^{\text{(i)}}  = 
		\begin{cases}
		\mathcal{O}(\sqrt{\phi^{i-1}\li/k}) \\
		\text{   for $i=1, 2, \ldots, \log_\phi(k/\log_2{k})$}\\
		\mathcal{O}(\sqrt{\li})  \\
		\text{   for $i=\log_\phi(k/\log_2{k})+1$};
		\end{cases}
		\end{align*}
		\item The {\textnormal{Number of Transmissions}} $\mathcal{N}_{DecCkt}$ is bounded by
		\begin{align*}
		\mathcal{N}_{DecCkt}^{\text{(i)}} = 
		\begin{cases}
		\Theta(k/\phi^{i-1}) \\
		\text{ for  $i=1, 2, \ldots, \log_\phi(k/\log_2{k})$}\\
		\Theta(\sqrt{k})  \\
		\text{ for $i=\log_\phi(k/\log_2{k})+1$};
		\end{cases}
		\end{align*}
		

		\item The {\textnormal{Bit-precision}} required in each communication between nodes is bounded by
		\begin{align*}
		\mathcal{B}_{DecCkt} = \Theta(1).
		\end{align*}
	\end{itemize}
\end{lemma}

For the clearing stage, we use SHO-FA~\citep{bakshi2012sho} with an appropriate parameter setting. The following theorem states the performance guarantees of SHO-FA. 

\begin{theorem}[SHO-FA~\citep{bakshi2012sho}]
	\label{thm:shofa}
	For the {\textnormal{Sparsity Model $(n,m,k)$}}, the SHO-FA decoding algorithm with {\em encoding matrix} $\mathcal{A}_{\textnormal{SHO-FA}}$ has the following properties:
	\begin{enumerate}
		\item For every {\textnormal{input vector}} $\ivt\in\real^{\li}$, with probability 1-$\mathcal{O}(1/\sqrt{k})$ over the choice of $\mathcal{A}_{\textnormal{SHO-FA}}$, the algorithm produces a {\textnormal{recovery vector}} $\rvt$ such that $||\ivt-\rvt||_1/||\ivt||_1 \leq 2^{-\precision}$.
		\item The {\em number of measurements} $m\leq 2ck+\sqrt{k}$, where $c$ is the {\em measurements constant}.
	\end{enumerate}
\end{theorem}



\begin{lemma}[Error Probability: SA]
	\label{lemma:sa}
	For the {\textnormal{Sparsity Model $(n,m,k)$}}, the decoding circuit $DecCkt$ for {\em SA} implemented on the {\textnormal{Implementation Model $(\rho,\mu)$}} satisfies the following properties:
	\begin{enumerate}
		\item \label{sa:1} There is a constant $\varphi>e$ such that if from $i$-th stage to $(i+1)$-th stage, the area of sub-circuit increases from $\det(\Lambda_i)$ to $\det(\Lambda_{i+1})=\varphi\det(\Lambda_i)$, then for each sub-circuit in the $(i+1)$-th stage, it contains $ \varphi^{i}\li/C_{\textnormal{SA}}k$ {\em output-nodes} where $C_{\textnormal{SA}}>0$ is a constant. 
		
		\item In the $(i+1)$-th stage, for any sub-circuit, denote $\{A_{j}\}_{j\in\mathcal{S}=\{1,2,\ldots, \varphi^{i}\li/C_{\textnormal{SA}}k\}}$ the set of events that the $j$-th {\em output-node} corresponds to a non-zero entry, we have  \[\Pr\left[\bigvee_{j\in\mathcal{S}}{\left(\bigwedge_{i\in\mathcal{S}\backslash\{j\}}{\urcorner{A_i}}{\bigwedge{A_j}}\right)}\right]\geq1-\frac{1}{\varphi}\] where $\mathcal{S}=\{1,2,\ldots, \varphi^{i}\li/C_{\textnormal{SA}}k\}$.
		\item An {\em average block error probability} $\ep=\mathcal{O}(1/\sqrt{k})$ is achievable with a fixed {\em precision} $\precision$ under the regime $m=\Theta(k)$ and the sub-linear regime $k= n^{1-\beta}$ where $\beta\in (0,1)$.
	\end{enumerate}
\end{lemma}

\begin{proof}
	For the first property~\eqref{sa:1}, note that the event $B_j=\bigwedge_{i\in\mathcal{S}\backslash\{j\}}{\urcorner{A_i}}{\bigwedge{A_j}}$ is the event that within the circuit, only the $j$-th \textit{output-node} corresponds to a non-zero entry in the \textit{input-vector} $\ivt$. Furthermore $\Pr{\left[B_i \bigwedge{B_j}\right]}=0$ for all possible $i,j$. Therefore by the chain rule
	\begin{align}
	\IEEEnonumber
	&\Pr\left[\bigvee_{j\in\mathcal{S}}{\left(\bigwedge_{i\in\mathcal{S}\backslash\{j\}}{\urcorner{A_i}}{\bigwedge{A_j}}\right)}\right]\\
	\label{lemma:sa:1}
	&= \sum_{j=1}^{|\mathcal{S}|}\Pr\left[A_j|\bigwedge_{i\in\mathcal{S}\backslash\{j\}}\urcorner{A_i}\right] \left[\prod_{i=1}^{|\mathcal{S}|-1}\left(1-\Pr\left[A_i|\bigwedge_{k<i}{\urcorner{A_k}}\right]\right)\right].
	\end{align}
	
	Next using~\eqref{lemma:sa:1} we find bounds on $\Pr\left[A_j|\bigwedge_{i\in\mathcal{S}\backslash\{j\}}\urcorner{A_i}\right]$ and $\Pr\left[A_i|\bigwedge_{k<i}{\urcorner{A_k}}\right]$. Note that after the $i$-th stage, each event $A_j$ satisfies $\Pr\left[A_j\right]\leq 1/|\mathcal{S}|$ and $A_j$ is mutually independent of all but at most $\varphi^{i-1}\li/C_{\textnormal{SA}}k$ other $A_j$'s and $e\cdot \frac{\varphi^{i-1}\li}{k|\mathcal{S}|}=\frac{e}{\varphi}\leq 1$ by choosing $\varphi\geq e$. Hence, by Lov{\`a}sz local lemma (see for example the textbook~\cite{alon2004probabilistic}), we have $\Pr\left[A_i|\bigwedge_{k<i}{\urcorner{A_k}}\right]\leq \frac{C_{SA}k}{\varphi^{i-1}\li}$. Further, as the sub-lattice in the $i+1$-th stage is chosen uniformly at random, then there is a constant $C'$ such that $\Pr\left[A_j|\bigwedge_{i\in\mathcal{S}\backslash\{j\}}\urcorner{A_i}\right]\geq \frac{\Pr\left[A_j\right]}{C'}=\frac{C_{SA}k}{C'\varphi^{i-1}\li}$. Thus from Equation~\eqref{lemma:sa:1},
	\begin{align}
		\IEEEnonumber
		&\Pr\left[\bigvee_{j\in\mathcal{S}}{\left(\bigwedge_{i\in\mathcal{S}\backslash\{j\}}{\urcorner{A_i}}{\bigwedge{A_j}}\right)}\right]\\
	\IEEEnonumber
	&\geq \frac{C_{\textnormal{SA}}k}{C'\varphi^{i-1}\li}|\mathcal{S}| \cdot \left(1-\frac{C_{\textnormal{SA}}k}{\varphi^{i-1}\li}\right)^{|\mathcal{S}|}\\
	\label{lemma:sa:2}
	&\geq \frac{C_{\textnormal{SA}}k}{C'\varphi^{i-1}\li}\cdot \frac{\varphi^{i-1}\li}{C_{\textnormal{SA}}k} \cdot \left(1-\frac{C_{\textnormal{SA}}k}{\varphi^{i-1}\li}\right)^{\varphi^{i-1}\li/C_{\textnormal{SA}}k}.
	\end{align}
	
	Taking limit with respect to $\li$ and using the Inequality~\eqref{lemma:sa:2} we have \[\Pr\left[\bigvee_{j\in\mathcal{S}}{\left(\bigwedge_{i\in\mathcal{S}\backslash\{j\}}{\urcorner{A_i}}{\bigwedge{A_j}}\right)}\right]\geq \frac{e}{C'}.\]
	
	Note that the above lower bound is constant across all stages. Letting $\varphi={C'}/{\left(C'-e\right)}$ for an appropriate $C'$ and applying concentration inequalities under the {\em Sparsity Model} $\left(\li,\lo,p\right)$, the probability $P_e^{1}(\textnormal{SA})$ of the event that after first $\log_\varphi(k/\log_2{k})$ stages more than $1/\sqrt{k}$ unsolved non-zero entries remain is upper bounded as 
	\begin{align}
	\IEEEnonumber
	&P_e^{1}(\textnormal{SA}) \\
	\IEEEnonumber
	&\leq C''\log_\varphi(k/\log_2{k}) \exp\left(k/{\varphi^{\log_\varphi(k/\log_2{k})}}\right)\\
	\label{lemma:sa:3}
	& \leq C''/\sqrt{k}
	\end{align}
	for some constant $C''>0$ since there is no intersection between sub-circuits at each stage.
	
	For the clearing stage, by Theorem~\ref{thm:shofa} proved in~\citep{bakshi2012sho}, we have if we use $\Theta(\sqrt{k})$ measurements, then the probability $P_e^{2}(\textnormal{SA})$ of the event that $||\ivt-\rvt||_1/||\ivt||_1 \geq 2^{-\precision}$ is $P_e^{2}(\textnormal{SA})=\mathcal{O}(1/\sqrt{k})$. Using union bound, we have $\ep\leq P_e^{1}(\textnormal{SA})+P_e^{2}(\textnormal{SA})$, implying that $\ep=\mathcal{O}(1/\sqrt{k})$.
\end{proof}

\begin{lemma}[Error Probability: CA]
	\label{lemma:ca}
	For the {\textnormal{Sparsity Model $(n,m,k)$}}, the decoding circuit $DecCkt$ for {\em CA} implemented on the {\textnormal{Implementation Model $(\rho,\mu)$}} achieves an average block error probability $\ep=\mathcal{O}(1/\sqrt{k})$ with a fixed precision $\precision$ under the regime $m=\Theta(k)$ and the sub-linear regime $k= n^{1-\beta}$ where $\beta\in (0,1)$.
\end{lemma}

\begin{proof}
	Using the same argument in Lemma~\ref{lemma:sa}, it suffices to show the probability $P_e^{1}(\textnormal{CA})$ of the event that more than $1/\sqrt{k}$ non-zero entries being left undecoded after first $\log_\phi(k/\log_2{k})$ stages satisfies $P_e^{1}(\textnormal{CA})=\mathcal{O}(1/\sqrt{k})$.
	
	Hence the only thing we need to show is for some choices of the \textit{encoding matrix} $\mathcal{A}_{\textnormal{CA}}$, the probability for at most $1/\sqrt{k}$ unsolved entries at the $\left(\log_\phi(k/\log_2{k})-1\right)$-th stage before the clearing stage is $P_e^{1}(\textnormal{CA})=\mathcal{O}(1/\sqrt{k})$. Therefore if at the $\left(\log_\phi(k/\log_2{k})-1\right)$-th stage, there is a constant fraction $\rho\geq 1-{1}/{\phi}$ of sub-circuits which contain at most $\log_\phi(k)$ non-zero entries (may be solved in the former stages), we then could claim that $P_e^{1}(\textnormal{CA})=\mathcal{O}(1/\sqrt{k})$ by using concentration inequalities and the fact that the sub-circuits at the $\left(\log_\phi(k/\log_2{k})-1\right)$-th stage have no intersection with one another. This is true because only $\rho\left(\log_\phi(k)\right)^2$ non-zero entries remain undecoded. Let the event that at the $\left(\log_\phi(k/\log_2{k})-1\right)$-th stage the $j$-th sub-circuit is of at most output-nodes corresponds to $\log_\phi(k)$ non-zero entries be $B_j$. Note that by the definition of our \textit{Sparsity Model ($\li$,$m$,$p$)} in Definition~\ref{pcs} the probability for each sub-circuit has at most $\log_\phi(k)$ non-zero entries is bounded from above by
	\begin{align*}
	\Pr\left[B_j\right] \leq \left(1-\frac{C_{\textnormal{CA}}\phi^{i}k}{\li}\right)^{\li/C_{\textnormal{CA}}\phi^{i}k}
	\end{align*}
	where $i=\log_\phi(k/\log_2{k})-1$.
	
	Taking limit with respect to $\li$, we have $\lim_{\li\rightarrow \infty}\Pr\left[\bar{B_j}\right] > 1-1/{e}$. Therefore by letting $\phi\leq e$ we have $P_e^{1}(\textnormal{CA})=\mathcal{O}(1/\sqrt{k})$ and hence using the same argument in Lemma~\ref{lemma:sa}, we conclude this lemma.
	
\end{proof}
Now we prove Theorem~\ref{thm:bitmeter} and Corollary~\ref{coro:tight} stated in Section~\ref{Main Results}.

\vspace{10pt}

\subsection{\em Proof of Theorem~\ref{thm:bitmeter} and Corollary~\ref{coro:tight}}



	From Lemma~\ref{lemma:prop}, the number of transmissions $\mathcal{N}^{\text{(i)}}_{DecCkt}$ decays geometrically. Combining this with Lemma~\ref{lemma:prop}, we conclude that the total {\textit{bit-meters}} are bounded by:
	\begin{align}
	\IEEEnonumber
		&\mu{(DecCkt)} \\
	\IEEEnonumber
	&= \sum_{i=1}^{\log_\phi(k/\log_2{k})+1}{\mathcal{D}_{DecCkt}^{\text{(i)}}\mathcal{N}_{DecCkt}^{\text{(i)}}\mathcal{B}_{DecCkt}}\\
	\label{noprecision}
	&=  \sum_{i=1}^{\log_\phi(k/\log_2{k})}{\mathcal{O}\left(\sqrt{\li k/\phi^{i-1}}\right)}+\mathcal{O}\left( \sqrt{nk}\right)\\
	\label{equation2_thm4}
	&=  \mathcal{O}\left( \sqrt{nk}\right)\\
	\label{equation3_thm4}
	&=  \mathcal{O}\left( \sqrt{\frac{nk}{\log{n}}}\sqrt{\log{\frac{1}{\ep}}}\right).
	\end{align}
	
	We get (\ref{noprecision}) because of the assumption that the precision parameter $\precision$ is fixed. Summing up all terms in (\ref{noprecision}) yields equation~(\ref{equation2_thm4}). By Theorem~\ref{thm:shofa}, Lemma~\ref{lemma:ca} and Lemma~\ref{lemma:sa}, the average block error probability $\ep$ satisfies $\ep=\mathcal{O}(1/\sqrt{k})$. Since $\mathcal{O}(\sqrt{\log\li})=\mathcal{O}(\sqrt{\log k})$ in the sub-linear regime $k= n^{1-\beta}$ where $\beta\in (0,1)$, we have $\mathcal{O}(\sqrt{\log\li})=\mathcal{O}(\sqrt{1/\log{\ep}})$ implying (\ref{equation3_thm4}). Therefore combining the above with Corollary~\ref{corollary}, we get $\mu{(DecCkt)} = \Theta\left( \sqrt{\frac{nk}{\log{n}}}\sqrt{\log{\frac{1}{\ep}}}\right)$.
\qed
\bibliographystyle{IEEEtran}
\bibliography{ref}

\end{document}